\documentclass[11pt]{article}

\usepackage{bm, commath}
\usepackage{natbib}
\usepackage{caption}
\usepackage{graphicx}
\usepackage{amsmath, amsfonts, amsthm}
\usepackage{float}
\usepackage{booktabs,siunitx}
\usepackage{url}
\usepackage{multirow}
\usepackage{amssymb}
\usepackage{blkarray}
\usepackage{bbm}
\usepackage{multicol}
\usepackage{authblk}
\usepackage[shortlabels]{enumitem}

\usepackage{makecell}

\usepackage{listings}
\usepackage{bbm}
\usepackage{textcomp}
\usepackage[margin=1in]{geometry}
\usepackage{authblk}
\usepackage[ruled]{algorithm2e}
\SetKwInput{KwParam}{Parameter}
\SetAlgoCaptionLayout{centerline}

\usepackage{sectsty}
\usepackage[doublespacing]{setspace}
\usepackage[utf8]{inputenc}
\usepackage{float}
\usepackage{tikz}
\usetikzlibrary{shapes,decorations,arrows,calc,arrows.meta,fit,positioning}

\def\bs{\boldsymbol}
\def\I{\mathbbm{1}}
\usepackage[textsize=scriptsize, textwidth = 2cm, shadow]{todonotes}

\newtheorem*{assumption*}{Assumption}

\newtheorem{proposition}{Proposition}

\newtheorem{innercustomgeneric}{\customgenericname}
\providecommand{\customgenericname}{}
\newcommand{\newcustomtheorem}[2]{%
  \newenvironment{#1}[1]
  {%
   \renewcommand\customgenericname{#2}%
   \renewcommand\theinnercustomgeneric{##1}%
   \innercustomgeneric
  }
  {\endinnercustomgeneric}
}

\newcustomtheorem{customthm}{Theorem}
\newcustomtheorem{customassumption}{Assumption}

\usepackage{mathtools}

\usepackage{xcolor}

\makeatletter
\renewcommand{\algocf@captiontext}[2]{#1\algocf@typo. \AlCapFnt{}#2} 
\def\@algocf@capt@plain{top}
\renewcommand{\algocf@makecaption}[2]{%
  \addtolength{\hsize}{\algomargin}%
  \sbox\@tempboxa{\algocf@captiontext{#1}{#2}}%
  \ifdim\wd\@tempboxa >\hsize
  \hskip .5\algomargin%
  \parbox[t]{\hsize}{\algocf@captiontext{#1}{#2}}
  \else%
  \global\@minipagefalse%
  \hbox to\hsize{\box\@tempboxa}
  \fi%
  \addtolength{\hsize}{-\algomargin}%
}
\makeatother

\usepackage[utf8]{inputenc}
\usepackage{float}
\usepackage{tikz}
\usetikzlibrary{positioning,chains,fit,shapes,calc}
\usepackage{multicol,multirow}
\usepackage{caption}
\usepackage{subfig}
\usepackage{graphicx}
\usepackage{multicol,multirow}
\usepackage{natbib}
\usepackage{bbm}
\usepackage{amsthm}
\def\bs{\boldsymbol}
\def\I{\mathbbm{1}}

\begin{document}
\sectionfont{\bfseries\large\sffamily}%
\subsectionfont{\bfseries\sffamily\normalsize}%

\title{The role of randomization inference in unraveling individual treatment effects in early phase vaccine trials}

\author[1]{Zhe Chen}
\author[2]{Xinran Li}
\author[3]{Bo Zhang\thanks{Vaccine and Infectious Disease Division, Fred Hutchinson Cancer Center, WA 98109, United States; Email: {\tt bzhang3@fredhutch.org}} }

\affil[1]{Department of Statistics, University of Illinois Urbana-Champaign}
\affil[2]{Department of Statistics, University of Chicago}
\affil[3]{Vaccine and Infectious Disease Division, Fred Hutchinson Cancer Center}

\date{}

\maketitle

\noindent
\textsf{{\bf Abstract}: Randomization inference is a powerful tool in early phase vaccine trials when estimating the causal effect of a regimen against a placebo or another regimen. Randomization-based inference often focuses on testing either Fisher's sharp null hypothesis of no treatment effect for any unit or Neyman's weak null hypothesis of no sample average treatment effect. Many recent efforts have explored conducting exact randomization-based inference for other summaries of the treatment effect profile, for instance, quantiles of the treatment effect distribution function. In this article, we systematically review methods that conduct exact, randomization-based inference for quantiles of individual treatment effects (ITEs) and extend some results to a special case where na\"ive participants are expected not to exhibit responses to highly specific endpoints. These methods are suitable for completely randomized trials, stratified completely randomized trials, and a matched study comparing two non-randomized arms from possibly different trials. We evaluate the usefulness of these methods using synthetic data in simulation studies. Finally, we apply these methods to HIV Vaccine Trials Network Study 086 (HVTN 086) and HVTN 205 and showcase a wide range of application scenarios of the methods. \textsf{R} code that replicates all analyses in this article can be found in first author's GitHub page at \url{https://github.com/Zhe-Chen-1999/ITE-Inference}.}%

\vspace{0.3 cm}

\section{Introduction}
\subsection{Early phase vaccine trials; vaccine-induced immune responses; heterogeneity}
\label{subsec: early phase vaccine trials}
One primary objective in study protocols of early phase clinical trials of experimental vaccines is to evaluate the vaccine-induced immunogenicity. Vaccine-induced immune responses are often heterogeneous among study participants. To illustrate, Figure \ref{fig: 086 heterogeneity} exhibits the observed serum IgG binding antibody multiplex assay (BAMA) responses to two antigens, Con 6 gp120/B and gp41, among study participants in a phase 1, multi-arm, placebo-controlled clinical trial conducted via the HIV Vaccine Trials Network (HVTN) \citep{ditse2020effect, huang2022baseline}. HVTN 086/SAAVI 103 (HVTN 086 henceforth) enrolled a total of $184$ participants into $4$ study arms; within each study arm, participants were randomized to a candidate vaccine regimen or placebo. It is transparent from Figure \ref{fig: 086 heterogeneity} that vaccine recipients' binding antibody responses ranged from ``potent and mostly homogeneous" (e.g., response to gp41 among recipients of regimens 1 and 4) to ``weak and heterogeneous" (e.g., response to Con 6 gp120/B among recipients of regimen 2). This within-participant heterogeneity in vaccine-induced immune responses has been well-documented in many vaccines, including those for Covid-19, influenza, dengue, and hepatitis B, \citep{huang2022baseline} and could at least partially explain the lack of efficacy in phase 2b/3, HIV-1 vaccine trials. 

\begin{figure}[ht]
    \centering
    \includegraphics[width = \textwidth]{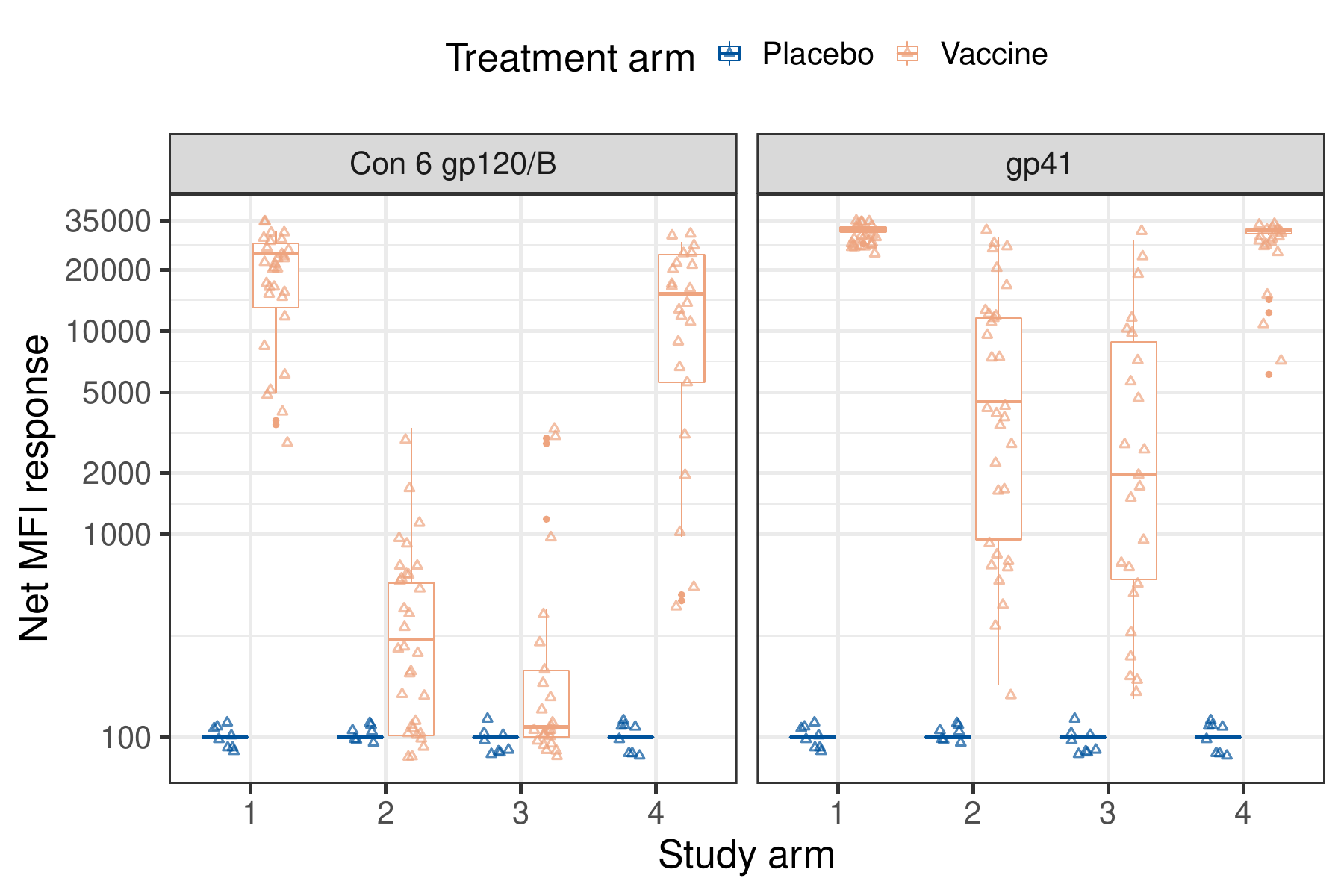}
    \caption{Observed serum IgG BAMA responses to Con6 gp120/B (left panel) and gp41 (right panel) among study participants in the HIV Vaccine Trials Network (HVTN) study 086. A total of $4$ study arms were plotted. Within each arm, participants were randomized to a vaccine regimen or placebo. A small perturbation was added to each observation to aid data visualization.}
    \label{fig: 086 heterogeneity}
\end{figure}

Researchers routinely characterize a vaccine regimen's induced immune responses (e.g., the BAMA responses in Figure \ref{fig: 086 heterogeneity}) by estimating and reporting its sample average treatment effect (SATE) against the placebo, which could be unbiasedly estimated in a randomized clinical trial, and assess and rank multiple regimens by comparing their estimated SATEs. When developing a challenging vaccine product like an HIV-1 vaccine, researchers have long realized that the response rates among study participants are often highly variable, and a significant proportion of participants could exhibit immune responses below the assay limit of detection (LOD) or lower limit of quantification (LLOQ). Hence, summarizing and comparing immune response profiles based on the mean difference alone could mask significant and perhaps meaningful heterogeneities. The current standard practice is to complement the estimated SATE and its 95\% confidence interval by further reporting (i) descriptive statistics and boxplots summarizing the spread of immune responses elicited by each regimen, (ii) the percentage of positive or high responders within each group, and (iii) the mean response among the subset of positive or high responders. Unfortunately, unlike the SATE which is a well-defined, albeit less than comprehensive, causal estimand, neither the descriptive statistics nor the mean difference in responses among positive or high responders constitutes a formal confidence statement about the ``treatment effect" of a regimen versus placebo or another regimen.

\subsection{Science table; estimands of interest; outline of the article}
\label{subsec: science table}

What is a well-defined causal estimand that captures treatment effect heterogeneity in early phase clinical trials? It is instructive to examine what \citet{rubin2005causal} refers to as a ``science table." Table \ref{tb: science table} summarizes the potential outcomes and unit-level treatment effects of $N_1 + N_0$ study participants in a clinical trial, where $N_1$ study participants are randomized to  Regimen $1$ and the other $N_0$ participants to Regimen $0$. In the science table, $Y_i(1)$ and $Y_i(0)$ denote study participant $i$'s potential immune responses of interest under two regimens, respectively, and only one of the two responses is observed depending on the actual regimen assigned to the participant. The contrast in the two potential outcomes, $\tau_i = Y_i(1) - Y_i(0)$, denotes the unit-level treatment effect \citep{rubin2005causal}. We will refer to $\tau_i$ as an \emph{individual treatment effect} (ITE) in this article \citep{caughey2021randomization, lipkovich2023overview}. The collection of ITEs, $\mathcal{T} = \{\tau_i,~i = 1, \dots, N_1 + N_0\}$, is the causal quantity of \emph{ultimate interest}, in the sense that any summary treatment effect can be derived from $\mathcal{T}$ (e.g., the sample average treatment effect $\overline{\tau}$). Let $\tau_{(i)}$ denote the $i$th largest treatment effect. The immunogenicity profile of a vaccine regimen (against placebo or a competing regimen), as revealed in an early phase clinical trial, is completely characterized by $\mathcal{T} = \{\tau_{(i)},~i = 1, \dots, N_1 + N_0\}$. Unfortunately, elements in $\mathcal{T}$ are almost never completely observed, so statistical inference is needed.

\begin{table}[ht]\centering
 \caption{Science table of $N_0 + N_1$ study participants. A total of $N_0$ are randomized to regimen 0 and the other $N_1$ are randomized to regimen $1$. Each participant is associated with two potential immune responses $Y_i(0)$ and $Y_i(1)$ corresponding to regimen $0$ and $1$, though only the one corresponding to the actual regimen assignment is observed (\textbf{boldface}). Each participant $i$ is associated with an individual treatment effect $\tau_i$.}
    \label{tb: science table}
    \resizebox{\textwidth}{!}{\begin{tabular}{cccccc}
    \hline
Participants &Regimen & $Y(1)$ & $Y(0)$ &Individual Treatment Effects\\ \hline
\\[-0.9em]
$1$  & Regimen 1 & $\mathbf{Y_{1}(1)}$ & $Y_{1}(0)$ & $\tau_1 = Y_1(1) - Y_1(0)$ \\
$2$  & Regimen 1 & $\mathbf{Y_{2}(1)}$ & $Y_{2}(0)$ & $\tau_2 = Y_2(1) - Y_2(0)$ \\
$\vdots$  &$\vdots$ &$\vdots$ &$\vdots$ &$\vdots$\\
$N_1$  & Regimen 1 & $\mathbf{Y_{N_1}(1)}$ & $Y_{N_1}(0)$ & $\tau_{N_1} = Y_{N_1}(1) - Y_{N_1}(0)$ \\
$N_1 + 1$  & Regimen 0 & $Y_{N_1 + 1}(1)$ & $\mathbf{Y_{N_1 + 1}(0)}$ &$\tau_{N_1 + 1} = Y_{N_1 + 1}(1) - Y_{N_1 + 1}(0)$  \\
$\vdots$ &$\vdots$ &$\vdots$ &$\vdots$ &$\vdots$\\
$N_1 + N_0$  & Regimen 0 & $Y_{N_1 + N_0}(1)$ & $\mathbf{Y_{N_1 + N_0}(0)}$ &$\tau_{N_1 + N_0} = Y_{N_1 + N_0}(1) - Y_{N_1 + N_0}(0)$ \\ \hline
\end{tabular}}
\end{table}

\subsection{Outline of the article}
\label{subsec: outline}

Our primary goal in this article is to provide a brief overview of two classical causal null hypotheses, Fisher's sharp null hypothesis and Neyman's weak null hypothesis, and introduce a novel class of null hypotheses regarding the quantiles of individual treatment effects. We then review some recently proposed methods that conduct exact, randomization-based inference for ITE quantiles. These methods are suitable for a range of practical scenarios, including completely randomized trials, block randomized trials, and a matched-pair design that compares two non-randomized vaccine arms or a new vaccine arm against historical controls. We argue that ITE quantiles, and more generally the distribution function of ITEs, offer a perspective that could complement usual estimands (e.g., sample average treatment effect) in early phase vaccine trials that aim to assess vaccine-induced immunogenicity. Finally, we present a comprehensive case study of the immunogenicity data derived from HIV Vaccine Trials Network (HVTN) Study 086 and explore how different methods may facilitate decision-making and improve the evaluation of vaccine regimens.

\section{Framework, notation, and different null hypotheses}
\subsection{Fisher's sharp null hypothesis}
\label{subsec: review Fisher sharp null}
We consider a two-arm randomized trial on $N$ study participants under Neyman-Rubin's potential outcomes framework \citep{neyman1923application,rubin1974estimating}. Let $Y_i(1)$ and $Y_i(0)$ denote participant $i$'s potential outcomes under \textsf{Regimen 1} and \textsf{Regimen 0} and $\tau_i = Y_i(1) - Y_i(0)$ participant $i$'s ITE.
We collect the set of $N$ potential outcomes under \textsf{Regimen 1} in $\mathbf{Y}(1) = (Y_1(1), Y_i(1), \dots, Y_N(1))^\intercal$, those under \textsf{Regimen 0} in $\mathbf{Y}(0) = (Y_1(0), Y_i(0), \dots, Y_N(0))^\intercal$, 
and the set of ITEs in $\boldsymbol{\tau} = (\tau_1, \tau_2, \dots, \tau_N)^\intercal$. 
Let $Z_i$ denote the treatment assignment to participant $i$ such that $Z_i = 1$ if participant $i$ is assigned \textsf{Regimen 1} and $0$ otherwise. The set of treatment assignments is collected in $\mathbf{Z} = (Z_1, Z_2, \dots, Z_N)^\intercal$. For each study participant $i$, the observed outcome $Y_i$ satisfies $Y_i = Z_i \cdot Y_i(1) + (1 - Z_i) \cdot Y_i(0)$. We collect $N$ observed outcomes in the vector $\mathbf{Y} = (Y_1, Y_2, \dots, Y_N)^\intercal$. In randomization inference, potential outcomes of study participants are viewed as fixed quantities and researchers rely \emph{solely} on the treatment assignment mechanism to draw valid causal conclusions. In other words, randomization forms what Sir Ronald Fisher referred to as the ``reasoned basis" for causal inference \citep{Fisher1935design}. In a completely randomized experiment (CRE), $N_1$ study participants are randomly assigned to \textsf{Regimen 1} and the other $N_0 = N - N_1$ to \textsf{Regimen 0}. In his seminal work \emph{Design of Experiments}, \cite{Fisher1935design} 
advocated testing the following \emph{sharp} null hypothesis:
\begin{equation*}
    H_{0, \text{sharp}}: \mathbf{Y}(1) = \mathbf{Y}(0), 
\end{equation*}
which states that the treatment has no effect whatsoever on any study participant, or equivalently $\boldsymbol{\tau} = 0$. 
Under this null hypothesis, the missing potential outcome of each study participant can then be imputed because $Y_i(1) = Y_i(0) = Y_i$, for $i = 1, \dots, N$, and any test statistic has its null distribution completely specified, leading to an \emph{exact} $p$-value upon comparing the observed value of the chosen test statistic to this null distribution. The Fisher randomization test (FRT) scheme works for a generic sharp null hypothesis of the following form:
\begin{equation*}
H_{\boldsymbol{\delta}}: \boldsymbol{\tau}=\boldsymbol{\delta},
\end{equation*}
where $\boldsymbol{\delta}=\left(\delta_1, \delta_2, \ldots, \delta_N\right)^{\intercal} \in \mathbb{R}^N$ is a prespecified  vector of constants.
 Let $t(\bs{z}, \bs{y})$ denote a generic function of the assignment vector $\bs{z} \in \{0,1\}^N$ and outcome vector $\bs{y}\in \mathbb{R}^N$. 
 Following \citet{rosenbaum2002observational}, we consider test statistics of the form $t\left(\boldsymbol{Z}, \boldsymbol{Y}_{\bs{Z}, \bs{\delta}}(0)\right)$, where $\bs{Z}$ is the vector of observed treatment assignment and $\bs{Y}_{\bs{Z}, \bs{\delta}}(0)$ is the vector of potential outcomes under \textsf{Regimen 0}. Given the observed data $\bs{Y}$ and under $H_{\boldsymbol{\delta}}$, we have $\boldsymbol{Y}_{\boldsymbol{Z}, \boldsymbol{\delta}}(0)=\boldsymbol{Y}-\boldsymbol{Z} \circ \boldsymbol{\delta}$, where $\circ$ represents element-wise multiplication.
In a CRE design, the null distribution of the test statistic has the following tail probability:
\begin{align}\label{eq:tail_prob}
    G_{\bs{Z}, \bs{\delta}}(c) & \equiv \Pr\{ t(\bs{A}, \bs{Y}_{\bs{Z}, \bs{\delta}}(0)) \ge c \} = \binom{N}{N_1}^{-1} \cdot \sum_{\bs{a}\in \{0,1\}^N: \sum_{i=1}^N a_i = N_1} \I\left\{ t(\bs{a}, \bs{Y}_{\bs{Z}, \bs{\delta}}(0)) \ge c\right\},
\end{align}
where $\bs{A}$ denotes a random treatment assignment vector under the CRE design, and the corresponding randomization-based, exact $p$-value is obtained by evaluating the tail probability in \eqref{eq:tail_prob} at the observed value of the test statistic: \begin{align}\label{eq:pval}
p_{\boldsymbol{Z}, \boldsymbol{\delta}} \equiv G_{\boldsymbol{Z}, \boldsymbol{\delta}}\left\{t\left(\boldsymbol{Z}, \boldsymbol{Y}_{\boldsymbol{Z}, \boldsymbol{\delta}}(0)\right)\right\}.
\end{align}

The sharp null $H_{0, \text{sharp}}$ specifies the entire ITE vector $\bs{\tau}$ and often equates it to a constant (e.g., $\bs{\tau} = c\bs{1}$). In some circumstances, treatment effect heterogeneity is likely to exist although its details (e.g., how the treatment effect varies across subgroups) may be unknown. In these cases, the sharp null hypothesis of no effect whatsoever, i.e., $H_{0, \text{sharp}}: \mathbf{Y}(1) = \mathbf{Y}(0)$, is still relevant and a meaningful first step of data analysis \citep{imbens2015causal}; however, the null hypothesis of a constant additive effect may not be of primary interest.

\subsection{Neyman's weak null hypothesis}
\label{subsec: Neyman weak null}
Testing a weak null hypothesis emerges as an alternative to testing a sharp null hypothesis. A weak null hypothesis only hypothesizes one aspect or property of the collection of ITEs (e.g., the \textit{average} treatment effect) without putting any restriction on each component ITE. Much of the literature on testing a weak null hypothesis, dating back to \citet{neyman1935statistical}, has focused almost exclusively on testing the following weak null hypothesis of no sample average treatment effect:
\begin{equation*}
    H_{0, \text{weak}}: \overline{\tau} = \frac{1}{N}\sum^N_{i = 1} \tau_i = 0.
\end{equation*}
\noindent Unlike the Fisher Randomization Test (FRT) that is \emph{exact}, randomization-based tests for SATE in general require large-sample approximation, although some of them also enjoy finite-sample validity for a certain sharp null hypothesis. 
For instance, \cite{ding2018randomization} showed that, under a CRE design, the FRT based on a studentized $t$-statistic for the sharp null hypothesis $H_{0, \text{sharp}}$ is also asymptotically valid for the weak null hypothesis $H_{0, \text{weak}}$; 
see also \cite{cohen2022gaussian} and \cite{wu2021randomization} for extensions to more general test statistics and experimental designs.

As discussed in Section \ref{subsec: early phase vaccine trials}, the SATE is one important, albeit not comprehensive, assessment of a regimen's treatment effect profile. Moreover, when the outcomes have heavy tails, the SATE could be sensitive to the outliers and the finite population asymptotic approximation tends to work poorly in these cases \citep{caughey2021randomization}. These aspects are particularly relevant in pre-clinical studies (e.g., nonhuman primates studies) and early phase clinical trials, where the sample size could be as small as $10$ to $20$ per arm and treatment effect heterogeneity is often expected.  

\subsection{Beyond the SATE: quantiles and proportions}
\label{subsec: quantile and proportion}
An exclusive focus on the SATE could arrive at unsatisfactory conclusions in face of large treatment effect heterogeneity.
For example, a treatment that is harmful to a majority of participants could still have a positive treatment effect on average, due to possibly some large ITEs on a small proportion of the cohort. A most comprehensive summary of ITE profile is its distribution function. Let $F(c) \equiv N^{-1} \sum_{i=1}^N \I(\tau_i \le c)$, for $c\in \mathbb{R}$, denote the distribution function of ITEs and $F^{-1}(\beta) \equiv \inf\{c: F(c) \ge \beta\}$, for $\beta \in (0,1]$, denote the corresponding quantile function. We focus on the $N$ study participants so that the quantile function can take at most $N$ values, i.e., the sorted ITEs $\tau_{(1)} \le \tau_{(2)} \le \ldots \le \tau_{(N)}$. 
More precisely, we have
$F^{-1}(\beta) = \tau_{(k)}$, with $k = \lceil N\beta \rceil$ denoting the ceiling of $N\beta$, for $\beta \in (0,1
]$. In addition, 
the distribution function can be equivalently written as $F(c) = 1 - N(c)/N$, for $c\in \mathbb{R}$, 
where $N(c) \equiv \sum_{i=1}^N \I(\tau_i > c)$ denotes the number of participants with ITEs exceeding a threshold $c$. 
We consider the following null hypothesis for any $0\le k \le N$ and $c\in \mathbb{R}$:
\begin{align}\label{eq:Hkc}
    H_{k, c}: \tau_{(k)} \le c 
    \ \Longleftrightarrow \ 
    N(c) \le N-k. 
\end{align}
For descriptive simplicity, we define $\tau_{(0)} = -\infty$. 
Here we focus on one-sided testing with alternatives favoring large treatment effects; by inverting the tests, we may then obtain lower confidence limits for ITE quantiles $F^{-1}(\beta)$'s (or equivalently $\tau_{(k)}$'s) and proportions of participants with ITEs exceeding any thresholds $1-F(c)$'s (or equivalently $N(c)$'s). 
To obtain one-sided tests with alternatives favoring small treatment effects, one may use the same procedure but with observed outcomes' signs flipped or the treatment/control status switched. Two-sided tests can be constructed by combining two one-sided tests using, say, the Bonferroni correction. 

\section{A review of randomization inference for ITE profiles}
\label{sec: review of methods}
\subsection{Completely randomized experiments}
\label{subsec: CRE}
In a recent article, \citet{caughey2021randomization} extend the FRT to testing the null hypothesis $H_{k,c}$ in \eqref{eq:Hkc}. Because $H_{k,c}$ is a composite null hypothesis and permits infinitely many imputation schemes, FRT is not directly applicable. Nevertheless, a valid $p$-value for testing $H_{k,c}$ can be obtained by maximizing the randomization $p$-value $p_{\boldsymbol{Z}, \boldsymbol{\delta}}$ in \eqref{eq:pval} over $\boldsymbol{\delta} \in \mathcal{H}_{k, c}$, where $\mathcal{H}_{k, c}$ denotes
 the set of vectors whose elements of rank $k$ are bounded by $c$,
 i.e., $\mathcal{H}_{k, c} \equiv \left\{ \boldsymbol{\delta} \in \mathbb{R}^N: \delta_{(k)}\leq c \right\} \subset \mathbb{R}^N$.
However, optimizing $\sup _{\boldsymbol{\delta} \in \mathcal{H}_{k, c}} p_{\boldsymbol{Z}, \boldsymbol{\delta}}$ is computationally challenging and can even be NP hard.
To address this challenge, \citet{caughey2021randomization} propose to use the class of rank score statistics of the following form:
\begin{equation}\label{def: rank score stat}
    t(\boldsymbol{z}, \boldsymbol{y})=\sum_{i=1}^{N} z_{i} \phi\left\{\mathrm{r}_{i}(\boldsymbol{y})\right\},
\end{equation}
where $\phi\{\cdot\}$ is a monotone increasing function, and $\mathrm{r}_i(\bs{y})$ denotes the rank of the $i$th coordinate of $\bs{y}$ using index ordering to break ties, 
assuming that the ordering has been randomly permuted before the analysis. 

Under a CRE design, the rank score statistic $t(\cdot, \cdot)$ defined in \eqref{def: rank score stat} is distribution free, in
the sense that for any $\boldsymbol{y}, \boldsymbol{y'} \in \mathbb{R}^{N}$, $t(\boldsymbol{Z},\boldsymbol{y})$ and $t(\boldsymbol{Z},\boldsymbol{y'})$ follow the same distribution.
 Because of this distribution free property, the imputed randomization distribution in \eqref{eq:tail_prob} reduces to a distribution that does not depend on the observed treatment assignment $\boldsymbol{Z}$ or the hypothesized treatment effect $\boldsymbol{\delta}$, i.e.,
\begin{equation}\label{def: G}
G_{\boldsymbol{Z}, \boldsymbol{\delta}}(c) \equiv \operatorname{Pr}\left\{t\left(\boldsymbol{A}, \boldsymbol{Y}_{\boldsymbol{Z}, \boldsymbol{\delta}}(0)\right) \geq c\right\}=\operatorname{Pr}\{t(\boldsymbol{A}, \boldsymbol{y}) \geq c\} \equiv G(c),
\end{equation}
where $\boldsymbol{y}$ can be any constant vector in $\mathbb{R}^N$. 
Consequently, the valid $p$-value $\sup _{\boldsymbol{\delta} \in \mathcal{H}_{k, c}} p_{\boldsymbol{Z}, \boldsymbol{\delta}}$ for testing $H_{k,c}$ in \eqref{eq:Hkc} simplifies to
\begin{equation}
\label{eq: valid p-val for experimental units}
p_{k,c}^{\textup{R}} \equiv \sup _{\boldsymbol{\delta} \in \mathcal{H}_{k, c}} p_{\boldsymbol{Z}, \boldsymbol{\delta}} = \sup _{\boldsymbol{\delta} \in \mathcal{H}_{k, c}} G\{t(\boldsymbol{Z}, \boldsymbol{Y}-\boldsymbol{Z} \circ \boldsymbol{\delta})\}=G\left\{\inf _{\boldsymbol{\delta} \in \mathcal{H}_{k, c}} t(\boldsymbol{Z}, \boldsymbol{Y}-\boldsymbol{Z} \circ \boldsymbol{\delta})\right\},
\end{equation}
where the last equality holds because $G$ is monotone decreasing and $t(\boldsymbol{Z}, \boldsymbol{Y}-\boldsymbol{Z} \circ \boldsymbol{\delta})$ achieves its infimum over $\boldsymbol{\delta} \in \mathcal{H}_{k, c}$. Equation \eqref{eq: valid p-val for experimental units} suggests that, to obtain a valid $p$-value for testing $H_{k,c}$ based on a distribution free test statistic, it suffices to minimize the value of the test statistic $t\left(\boldsymbol{Z}, \boldsymbol{Y}_{\boldsymbol{Z}, \boldsymbol{\delta}}(0)\right)$ over $\boldsymbol{\delta} \in \mathcal{H}_{k, c}$.
 When using the rank score statistics in \eqref{def: rank score stat}, the infimum is achieved when the ranks of $Y_i(0)$'s of treated participants are minimized, or equivalently, treated participants' ITEs are maximized subject to $H_{k,c}$. 
 Because $Z_i = 0$ for those in the control group, their $Y_i(0)$'s do not directly contribute to the value of the test statistic.
 \citet{caughey2021randomization} show that the worst-case $p$-value corresponds to assigning arbitrarily large ITEs to the $N-k$ treated participants with the largest observed outcomes and $c$ to the remaining participants.
 Moreover, the inference is simultaneously valid in the sense that there is no need to conduct multiple testing correction when jointly testing many quantiles or the entire distribution function of ITEs.

There are two main sources of conservativeness in \citeauthor{caughey2021randomization}'s \citeyearpar{caughey2021randomization} approach.
First, the worst-case solution depends primarily on the potential outcomes under control for treated participants, that is, the term $\boldsymbol{Y}-\boldsymbol{Z} \circ \boldsymbol{\delta}$ for those with $Z_i = 1$ in the expression \eqref{eq: valid p-val for experimental units}. Second, the worst-case solution corresponds to assigning largest ITEs to the treated participants; this is unlikely by randomization because ITE can be understood as a pretreatment variable in a broad sense, and randomization tends to balance the distribution of pretreatment variables. More recently, \citet{CL} propose two enhanced methods that each tackle one limitation and achieve improved statistical power.

Their first method better leverages the information contained in the control participants by (A1) conducting level-$\alpha$ simultaneous inference for ITEs among treated participants, (A2) flipping treated and control participants and repeating the level-$\alpha$ simultaneous inference for control participants, and (A3) combining the intervals for all participants by ordering the one-sided intervals according to their lower confidence limits. The resulting ordered intervals can be shown to be simultaneously valid, level-$2\alpha$ confidence intervals for all $N$ ITEs. In this procedure, the choice of the test statistic could be different in Step (A1) and Step (A2). We will explore the choice of test statistics in some practical settings in the simulation study.
Their second method targets the second limitation in \citet{caughey2021randomization}. 
Instead of presuming that the largest $N-k$ ITEs are all among the treated participants, they view the number of participants with the largest $N-k$ effects in the treatment arm as a nuisance parameter and 
use \citeauthor{Berger1994}'s \citeyearpar{Berger1994} approach to control for the randomness of this nuisance. Their second method can be summarized as follows: (B1) apply the \citeauthor{Berger1994}'s \citeyearpar{Berger1994} correction to derive simultaneous confidence intervals for all participants; (B2) flip the role of treated and control participants and repeat Step (B1); and (B3) combine the intervals derived from Step (B1) and Step (B2) using the Bonferroni method. \citet{CL} show via simulations that both approaches deliver more powerful statistical inference compared to the original method in \citet{caughey2021randomization} while remaining exact in finite sample.


\subsection{Stratified randomized experiments}
\label{subsec: stratified experiment}
\citet{10.1093/biomet/asad030} extend the approach in \citet{caughey2021randomization} to stratified completely randomized experiments (SCRE). They consider testing the same null hypothesis $H_{k, c}$ in a study design with $S$ strata. Each stratum consists of $n_s$ study participants such that $N=\sum_{s=1}^S n_s$. \citet{10.1093/biomet/asad030} consider using the following stratified rank sum statistic: 
\begin{equation}
    \label{eq: stratified rank sum}
    t_{\text{str}}(\bs{z}, \bs{y})=\sum_{s=1}^S t_s\left(\bs{z_s}, \bs{y_s}\right)=\sum_{s=1}^S \sum_{i=1}^{n_s} z_{s i} \phi_s\left\{\operatorname{r}_i\left(\bs{y_s}\right)\right\},
\end{equation}
where $\phi_s\{\cdot\}$ again denotes some monotone increasing rank transformation for stratum $S = s$, $\bs{z}_s$ and $\bs{y}_s$ denote the subvectors of $\bs{z}$ and $\bs{y}$ corresponding to stratum $S = s$, 
and $\operatorname{r}_i(\bs{y_s})$ denotes the rank of the $i$th coordinate of $\bs{y_s}$ among all coordinates of $\bs{y_s}$. Again, the stratified rank sum statistic enjoys the distribution free property in a stratified completely randomized experiment, so that a valid $p$-value for testing $H_{k, c}$ has an equivalent form as in \eqref{eq: valid p-val for experimental units}: 
the $p$-value $p_{k,c}^{\textup{R}}$ defined as in \eqref{eq: valid p-val for experimental units}, but with $t(\cdot, \cdot)$ replaced by the stratified rank sum statistic $t_{\text{str}}(\cdot, \cdot)$ and $G(\cdot)$ being the tail probability of $t_{\text{str}}(\bs{Z}, \bs{y})$ for any $\bs{y}\in \mathbb{R}^N$, 
is a valid $p$-value for testing $H_{k,c}$. \citet{10.1093/biomet/asad030} then demonstrate that minimizing the test statistic over all possible ITEs compatible with the null hypothesis $H_{k,c}$ can be transformed into a multiple-choice knapsack problem, which can be solved efficiently and exactly using dynamic programming, or in a slightly conservative manner using a greedy algorithm. Implementation of the methods can be found in the \textsf{R} package \textsf{QIoT} \citep{10.1093/biomet/asad030}.
The first enhanced method in Section \ref{subsec: CRE} can also be extended to the SCRE \citep{CL}.

\subsection{Integrated analysis of non-randomized arms}
\label{subsec: integrated matched analysis}
In many circumstances, researchers may be interested in conducting a pooled analysis of data derived from multiple clinical trials. This can happen in two circumstances. First, researchers may be interested in comparing immunogenicity of two vaccine regimens, one from a current trial and the other from a historical trial. Second, researchers may be interested in comparing a vaccine regimen to historical controls. Because of a lack of randomization, a na\"ive comparison of outcomes from non-randomized arms assuming randomization could lead to a bias in estimating the treatment effects. In these studies, it is essential to adjust for baseline covariates that could predict immunogenicity; for instance, \citet{huang2022baseline} report that host baseline characteristics could predict a high-level binding antibody response with a cross-validated area under the receiver operating characteristics curve (AUROC) equal to $0.72$ (95\% CI: [0.68 0.76]). Many methods can be used to perform covariance adjustment; in early phase trials with small samples, one reasonable strategy is to conduct a matched cohort study \citep{rosenbaum2002observational,rosenbaum2010design}. One popular downstream data analysis strategy in a matched cohort study is to conduct randomization inference. Because of the matched-pair or matched-set structure, one needs to conduct stratified randomization inference discussed in Section \ref{subsec: stratified experiment}. We will illustrate the proposed workflow, from statistical matching to stratified randomization inference for the ITE profile, in a case study in Section \ref{sec: case study}.

\subsection{Relaxing the randomization assumption}
\label{subsec: relax RI}
In a matched analysis of immunogenicity data derived from non-randomized arms, the key assumption is a version of the ignorability assumption, which effectively says that within the strata defined by observed covariates being matched on, selection into a particular trial or study arm is randomized \citep{stuart2011use,dahabreh2019generalizing}. In these analyses, it is essential to examine the consequences of deviating from the randomization assumption, as \citet[Chapter 6]{rosenbaum2017observation} put it: ``the absence of an obvious reason to think that two groups are different falls well short of a compelling reason to think they are the same." In early phase vaccine trials, it is conceivable that healthier study participants with potentially stronger immune responses may be preferentially enrolled a certain trial or study arm, for instance, because of the difference in the vaccine dosing schedules (bolus vs. fractional dosing), and adjusting for observed covariates may not be sufficient in removing the selection bias.

\citet{10.1093/biomet/asad030} study a relaxed randomization inference scheme, where in lieu of assuming randomization within each stratum, the treatment assignment probability is allowed to deviate from randomization under a model that controls the maximum level of deviation \citep{rosenbaum1987sensitivity, rosenbaum2002observational}:
\begin{equation}
\label{eq: Gamma model}
    \frac{1}{\Gamma} \leq \frac{\pi_{sj}/(1-\pi_{sj})}{\pi_{sj'}/(1-\pi_{sj'})} \leq {\Gamma}, \quad 1 \leq j, j' \leq n_s, ~1 \leq s \leq S, 
\end{equation}
where $\pi_{sj}$ and $\pi_{sj'}$ denote the treatment assignment probabilities of participant $j$ and $j'$ in stratum $S = s$, respectively, and the sensitivity parameter $\Gamma \geq 1$ specifies the maximum odds ratio across all strata. Inferring the treatment effect under a biased randomization scheme is referred to as a sensitivity analysis in the analysis of non-randomized or observational data, and the goal is to investigate the maximum degree of deviation from the randomization assumption when a certain null hypothesis can no longer be rejected. Such sensitivity analysis methods have been developed for Fisher's sharp null hypothesis \citep{rosenbaum1987sensitivity, rosenbaum2002observational} and Neyman's weak null hypothesis of no sample average treatment effect \citep{fogarty2020studentized}; \citet{10.1093/biomet/asad030} generalize the method to testing the null hypothesis $H_{k, c}$ under a stratified CRE design.

\section{Placebo-controlled trials with highly specific endpoints}
\label{sec: treatment against placebo}

In a placebo-controlled vaccine trial, participants' potential outcomes under the placebo are often known \emph{a priori}, and these controls are nonetheless included in the study primarily for blinding purposes (e.g., preventing treatment arm information from being revealed to lab technicians). For instance, in many early phase HIV vaccine trials, the endpoints of interest are vaccine-induced, antigen-specific immune responses, and healthy and na\"ive study participants receiving the placebo are not expected to exhibit any of these highly specific immune responses to the antigens. Hence, we often have auxiliary information $Y_i(0) = \textsf{LOD}$ for all $i = 1, \cdots, N$, where $\textsf{LOD}$ represents an assay-specific limit of detection. In this stylistic case, we immediately know the ITEs of $N_1$ treated participants, based on which we may infer the entire ITE distribution. In this section, we formally derive a level-$\alpha$ confidence interval for a single ITE quantile $\tau_{(k)}$, prove the equivalence between our result and earlier results by \citet{sedransk1978confidence} in the context of simple random sampling, and derive simultaneously valid confidence intervals for multiple ITE quantiles or the entire ITE distribution.

Without loss of generality, 
we assume $\textsf{LOD} = 0$ so $Y_i(0) = 0$ for all $i$'s. Consider the null hypothesis $H_{k,c}$ in \eqref{eq:Hkc} for any $1\le k \le N$ and $c\in \mathbb{R}$. 
Under the CRE design, treated participants are a simple random sample of size $N_1$ from a total of $N$ participants.
Therefore, 
$n(c) \equiv
\sum_{i=1}^N Z_i \I(Y_i > c)
= 
\sum_{i=1}^N Z_i \I(\tau_i > c) 
$
follows 
a Hypergeometric distribution with parameters $(N, N(c), N_1)$.
The Hypergeometric distribution with parameters $(N, n, N_1)$ becomes stochastically larger as $n$ increases; that is, if $X$ follows a Hypergeometric distribution with parameters $(N, n, N_1)$,  $Y$ follows a Hypergeometric distribution with parameters $(N, n', N_1)$ such that $n'\geq n$, then $\mathbb{P}\left(X \leq t\right) \geq \mathbb{P}\left(Y \leq t\right)$ for every $t\in \mathbb{R}$. Together, these facts imply a finite-sample valid $p$-value for testing $H_{k,c}$, as summarized in the following proposition.  

\begin{proposition}\label{prop:hyper_pval}
For any $1\le k \le N$ and $c\in \mathbb{R}$, 
$p_{k, c}^{\textup{H}} \equiv G_{\textup{H}}(n(c); N, N-k, N_1)$ 
is a valid $p$-value for testing the null hypothesis $H_{k,c}$ in \eqref{eq:Hkc}, where $G_{\textup{H}}(x; N, n, N_1) \equiv \Pr(X\ge x)$ denotes the tail probability of a Hypergeometric random variable $X$ with parameters $(N, n, N_1)$. 
Specifically, under $H_{k,c}$, 
$\Pr(p_{k, c}^{\textup{H}} \le \alpha) \le \alpha$ for any $\alpha\in (0,1)$. 
\end{proposition}

\begin{proof}
    All proofs in the article can be found in Supplemental Material A.
\end{proof}

From Proposition \ref{prop:hyper_pval}, 
we can then conduct Lehmann-style test inversion to construct confidence intervals for $\tau_{(k)}$ and $n(c)$, 
for any $1\le k \le N$ and $c\in \mathbb{R}$. 
Moreover, 
due to the monotonicity of the $p$-value $p_{k, c}^{\textup{H}}$ in $k$ and $c$, 
the resulting confidence sets are intervals and have simpler forms that facilitate their computation. Proposition \ref{prop:hyper_ci} summarizes these results. For descriptive convenience, we let $y_{(1)} \le y_{(2)} \le \ldots \le y_{(N_1)}$ denote the sorted observed outcomes for treated participants and further define $y_{(0)} = -\infty$. 

\begin{proposition}\label{prop:hyper_ci}
\begin{itemize}
    \item[(a)] 
    $p_{k, c}^{\textup{H}}$ is monotone increasing in $c$ and decreasing in $k$. 

    \item[(b)] 
    For any $1\le k \le N$ and $\alpha\in (0,1)$, 
    $\{c: p_{k, c}^{\textup{H}} > \alpha, c \in \mathbb{R} \}$ is a $1-\alpha$ confidence interval for $\tau_{(k)}$, 
    and it can be equivalently written as $[y_{(k(\alpha))}, \infty)$ 
    with 
  $k(\alpha) \equiv  N_1 - Q_{\textup{H}}(1-\alpha; N, N-k, N_1)$, where $Q_{\textup{H}}(\theta; N, n, N_1)$ denotes the $\theta$-th quantile function of a Hypergeometric random variable $X$ with parameters $(N, n, N_1)$.

    \item[(c)] For any $c\in \mathbb{R}$ and $\alpha\in (0,1)$, 
    $\{N-k: p_{k, c}^{\textup{H}} > \alpha,  0 \le k \le N \}$ is a $1-\alpha$ confidence set for $N(c)$, 
    and it can be equivalently written as 
    $\{n_{c, \alpha}, n_{c, \alpha}+1, \ldots, N\}$ with $n_{c, \alpha}=N - \max \{k: G_{\textup{H}}(n(c) ; N, N-k, N_1) > \alpha, 0 \le k \le N\}$. 
\end{itemize}
\end{proposition}

Confidence interval results for $\tau_{(k)}$ can also be derived by applying results in \citet{sedransk1978confidence} who study inference for quantiles of a finite population based on simple random sampling. In Supplemental Material A.3, we prove the equivalence between the confidence interval results in Proposition \ref{prop:hyper_ci} and those derived by directly applying results in \citet{sedransk1978confidence}. In Supplemental Material A.3, we also prove that the confidence interval results for $N(c)$ in Proposition \ref{prop:hyper_ci} are equivalent to results in \citet{Wang:2015} who study inference for Hypergeometric distribution parameters. Hence, by the optimality results in \citet{Wang:2015}, confidence intervals for $N(c)$ in Proposition \ref{prop:hyper_ci}(c) are optimal one-sided confidence intervals. 

Evaluating the ITE profile of a treatment regimen goes beyond constructing confidence intervals for each individual ITE. Proposition \ref{prop:simul_ci} derives the simultaneous confidence intervals for multiple ITE quantiles, which provides a valid summary of treatment effect profile and is of primary interest in practice.

\begin{proposition}\label{prop:simul_ci} 
Consider a CRE design and suppose that we are interested in multiple quantiles of ITEs: $\tau_{(k_1)}, \tau_{(k_2)}, \ldots, \tau_{(k_J)}$, where $1\le k_1 \le \ldots \le k_J \le N$. 
Under a CRE design,
the simultaneous coverage probability for quantiles of ITEs is 
\begin{align}\label{eq:joint_ci_Y0_constant}
    \Pr\left( \bigcap_{j=1}^J \{ \tau_{(k_j)} \ge y_{(k_j(\alpha))} \} \right )
   & \ge  
    1- \Pr\left( \bigcup_{j=1}^J 
    \left\{ 
    \sum_{i=1}^N Z_i \I(i > k_j) > Q_{\textup{H}}(1-\alpha; N, N-k_j, N_1)
    \right\} \right ),
\end{align}
where the equality holds when all ITEs $\tau_i$'s are distinct. 
Recall that $y_{(k_j(\alpha))}$ is the  observed outcome of rank $k_j(\alpha) \equiv  N_1 - Q_{\textup{H}}(1-\alpha; N, N-k_j, N_1)$ in the treatment group.
\end{proposition}

The lower bound of the simultaneous coverage probability in \eqref{eq:joint_ci_Y0_constant} is sharp and obtained when all ITEs are distinct across all $N$ participants. 
More importantly, the lower bound does not depend on any unknown quantities and can be efficiently approximated using the Monte Carlo method.

Technically, the assay limit of detection specifies the lowest level of immune response to be detected and merely places an upper bound on $Y_i(0)$ rather than specifies its precise value. Proposition \ref{prop:Y(0)<=0} extends previous results and suggests that $p$-values and confidence statements derived from Propositions \ref{prop:hyper_pval}-\ref{prop:simul_ci} remain valid when $Y_i(0) \le 0$ rather than $Y_i(0) = 0$.

\begin{proposition}\label{prop:Y(0)<=0}
Propositions 1-3 hold when $Y_i(0) \le 0$.
\end{proposition}

\section{Simulation}
\label{sec: simulation}
\subsection{Goal and structure; methods to be compared; measurement of success}
\label{subsec: goal and structure}
The primary goal of this simulation study is to assess and compare the power of several competing methodologies reviewed in Section \ref{sec: review of methods} when comparing two vaccine regimens. To provide more relevant guidance for practice, we will use the immune response data generated from HVTN 086 as a basis for the data generating process. Specifically, we considered the following two data generating processes:

\begin{itemize}
     \item \textbf{DGP I:} We sampled BAMA responses of size $N_1$ against Con 6 gp120/B with replacement from the vaccine regimen \textsf{T1} and of size $N_0$ against Con 6 gp120/B with replacement from the vaccine regimen \textsf{T2}. 
      \item \textbf{DGP II:} We sampled BAMA responses of size $N_1$ against gp41 with replacement from the vaccine regimen \textsf{T1} and data of size $N_0$ against gp41 with replacement from the vaccine regimen \textsf{T2}.
\end{itemize}

\noindent For each of the $N = N_1 + N_0$ sampled data points, a Gaussian noise $\epsilon \sim N(0, 0.15^2)$ was added to their $\log_{10}$-transformed scale. \textsf{DGP I} represents a scenario where two vaccine regimens have similar spread in their immune responses and the treatment effects appear homogeneous. \textsf{DGP II} represents a distinct scenario where immune responses from two vaccine regimens have rather different spread and a heterogeneous ITE profile among vaccine recipients seems more plausible. We considered balanced design with the following three sets of sample sizes: $N_1 = N_0 = 30$, $N_1 = N_0 = 50$, and $N_1 = N_0 = 100$. 

For each data generating process, we performed randomization inference for the distribution function of ITEs (and equivalently the proportion of units with ITEs exceeding different thresholds). We constructed (i) individual confidence intervals for selected quantiles; and (ii) simultaneous confidence intervals for selected quantiles using the following methods and choices of test statistics: 

\begin{itemize}
    \item \textsf{M1:} The original method in \citet{caughey2021randomization}. We considered using Stephenson rank sum test with $s = 2$ (\textsf{M1-S2}) and $s = 6$ (\textsf{M1-S6});

    \item \textsf{M2:} The first enhanced method in \citet{CL} which combines inference for treated and control participants. This method involves choosing a test statistic when making inference for the treated and a second test statistic when making inference for the control. We considered the following $2 \times 2 = 4$ combinations of test statistics: Stephenson rank sum test with $s = 2$ or $s = 6$ when inferring ITEs for the treated in Step (A1) and Stephenson rank sum test with $s = 2$ or $s = 6$ when inferring ITEs for the control in Step (A2). These $4$ methods are referred to as \textsf{M2-S2-S2}, \textsf{M2-S2-S6}, \textsf{M2-S6-S2}, and \textsf{M2-S6-S6}. For instance, \textsf{M2-S2-S6} represents the method that uses Stephenson rank sum test with $s = 2$ in Step (A1) and with $s = 6$ in Step (A2).

    \item \textsf{M3:} The second enhanced method in \citet{CL} which considers a probabilistic allocation of the worst-case ITEs. This method also involves making inference twice, once using the raw data and a second time using data with flipped treatment assignments and negated outcomes, before combining the inference. Analogous to \textsf{M2}, we considered the following $2 \times 2 = 4$ combinations of test statistics: Stephenson rank sum test with $s = 2$ or $s = 6$ when inferring ITEs using raw data in Step (B1) and Stephenson rank sum test with $s = 2$ or $s = 6$ when inferring ITEs using data with flipped treatment assignments and negated outcomes in Step (B2). These $4$ methods are referred to as \textsf{M3-S2-S2}, \textsf{M3-S2-S6}, \textsf{M3-S6-S2}, and \textsf{M3-S6-S6}. 
\end{itemize}

For each data generating process, we repeated simulation $1,000$ times and measured the success of each method in recovering the ITE profile according to the following two criteria. First, for a selected quantile $k$ and the target estimand $\tau_{(k)}$, we reported the median of the $95\%$ one-sided lower confidence limits across $1,000$ simulations for each method. For a fixed $k$, a larger median value corresponds to a method being more powerful. Second, we compared the simultaneous confidence intervals to the true ITE profile $\mathcal{T} = \{\tau_{(i)},~i = 1, \dots, N\}$ by calculating $\textsf{SS} = |\mathcal{K}|^{-1}\sum_{k \in \mathcal{K}} (\hat{L}_{(k)} - \tau_{(k)})^2$, where $\mathcal{K}$ is a collection of quantiles of interest, and $\hat{L}_{(k)}$ denotes the lower confidence limit delivered by the method. For non-informative lower confidence limits, we set $\hat{L}_{(k)}$ to $-10$. We considered $\mathcal{K} = \{\lceil 0.5 \times N \rceil, \lceil 0.75 \times N \rceil, \lceil 0.8 \times N \rceil, \lceil 0.85 \times N \rceil, \lceil 0.9 \times N \rceil, \lceil 0.95 \times N \rceil\}$ and reported $\textsf{SS}$ averaged over $1,000$ simulations. A smaller $\textsf{SS}$ value corresponds to a method being more powerful and preferable. 

\subsection{Simulation results}
\label{subsec: simu results}

Figures \ref{fig: median 30} and \ref{fig: median 100} plot the median of the $95\%$ lower confidence limits for selected ITE quantiles when $N_1 = N_0 = 30$ and $N_1 = N_0 = 100$, respectively. Analogous plot when $N_1 = N_0 = 50$ can be found in Supplemental Material B.
For each fixed quantile, the red triangle represents the maximum median across different methods.
Table \ref{table: SS} further summarizes the average \textsf{SS} of each method under different data generating processes. 

We identify several consistent trends from the simulation results. First, each of the three methods exhibits improved power, as reflected by smaller average $\textsf{SS}$, as the sample size $N = N_1 + N_0$ increases. For instance, the average \textsf{SS} drops from $1.09$ to $0.78$ when $N$ increases from $60$ to $200$.
Second, methods \textsf{M2} and \textsf{M3} in general largely outperform \textsf{M1} in both pointwise and simultaneous inference. In fact, the gain of \textsf{M2} and \textsf{M3} over \textsf{M1} in relatively smaller quantiles like $50\%$ and $75\%$ quantiles can be quite significant. Third, unlike \textsf{M1} and \textsf{M2} whose pointwise and simultaneous inference coincides, \textsf{M3} suffers from multiple testing correction although the loss in power is negligible when the sample size is only moderately large (e.g., $N_1 = N_0 = 50$). For instance, the average \textsf{SS} for \textsf{M3-S2-S6} increases from $1.05$ to $1.10$ when $N_1 = N_0 = 50$.

When we examine different choices of test statistics for \textsf{M2}, methods \textsf{M2-S2-S6} and \textsf{M2-S6-S6} have similar performance and both outperform \textsf{M2-S2-S2} and \textsf{M2-S6-S2} for both DGP I and DGP II. Both \textsf{M2-S2-S6} and \textsf{M2-S6-S6} use the Stephenson rank sum test with a large $s$, e.g., $s = 6$, when the treatment status is flipped (i.e., Step A2 in \textsf{M2} and Step B2 in \textsf{M3}). This is consistent with the observation that a large $s$ is preferred when treated participants (or in this case, control participants before flipping the treatment status) have right-skewed outcomes and/or many large outliers \citep{caughey2021randomization}. Similar observation stands for method \textsf{M3}.

\begin{table}[ht]
\resizebox{\textwidth}{!}{
\centering
\begin{tabular}{lrrlrrrrrrrrrr}
   \hline
DGP & $N_1$ & $N_0$ & Inference & M1-S2 & M1-S6 & M2-S2-S2 & M2-S2-S6 & M2-S6-S2 & M2-S6-S6 & M3-S2-S2 & M3-S2-S6 & M3-S6-S2 & M3-S6-S6 \\ 
  \hline
DGP I &  30 &  30 & pointwise & 73.41 & 24.27 & 23.96 & 1.09 & 1.26 & 0.96 & 48.84 & 2.23 & 2.63 & 1.36 \\ 
  DGP I &  30 &  30 & simultaneous & 73.41 & 24.27 & 23.96 & 1.09 & 1.26 & 0.96 & 49.26 & 24.10 & 24.25 & 24.05 \\ 
  DGP I &  50 &  50 & pointwise & 49.10 & 24.19 & 23.88 & 0.92 & 1.08 & 0.87 & 24.30 & 1.05 & 1.21 & 1.02 \\ 
  DGP I &  50 &  50 & simultaneous & 49.09 & 24.19 & 23.88 & 0.92 & 1.07 & 0.87 & 24.47 & 1.10 & 1.26 & 1.07 \\ 
  DGP I & 100 & 100 & pointwise & 48.78 & 24.10 & 23.81 & 0.78 & 0.92 & 0.76 & 24.01 & 0.82 & 0.96 & 0.81 \\ 
  DGP I & 100 & 100 & simultaneous & 48.78 & 24.10 & 23.81 & 0.78 & 0.92 & 0.76 & 24.07 & 0.87 & 1.00 & 0.86 \\ 
  DGP II &  30 &  30 & pointwise & 64.92 & 21.67 & 21.15 & 1.31 & 2.01 & 1.37 & 43.09 & 2.30 & 2.74 & 1.39 \\ 
  DGP II &  30 &  30 & simultaneous & 64.92 & 21.67 & 21.16 & 1.31 & 2.01 & 1.37 & 43.36 & 20.81 & 21.59 & 20.81 \\ 
  DGP II &  50 &  50 & pointwise & 43.19 & 21.62 & 21.06 & 1.18 & 1.90 & 1.27 & 21.44 & 1.14 & 1.97 & 1.14 \\ 
  DGP II &  50 &  50 & simultaneous & 43.16 & 21.62 & 21.06 & 1.18 & 1.90 & 1.27 & 21.56 & 1.20 & 2.03 & 1.20 \\ 
  DGP II & 100 & 100 & pointwise & 43.04 & 21.57 & 20.98 & 1.05 & 1.79 & 1.17 & 21.19 & 0.92 & 1.77 & 0.92 \\ 
  DGP II & 100 & 100 & simultaneous & 43.04 & 21.57 & 20.98 & 1.05 & 1.79 & 1.17 & 21.25 & 0.97 & 1.82 & 0.97 \\ 
   \hline
\end{tabular}}
\caption{$\textsf{SS}$ averaged over $1,000$ simulations derived from each method for different data generating processes.}
\label{table: SS}
\end{table}

\begin{figure}[ht]
    \centering
    \includegraphics[width = 0.8\textwidth]{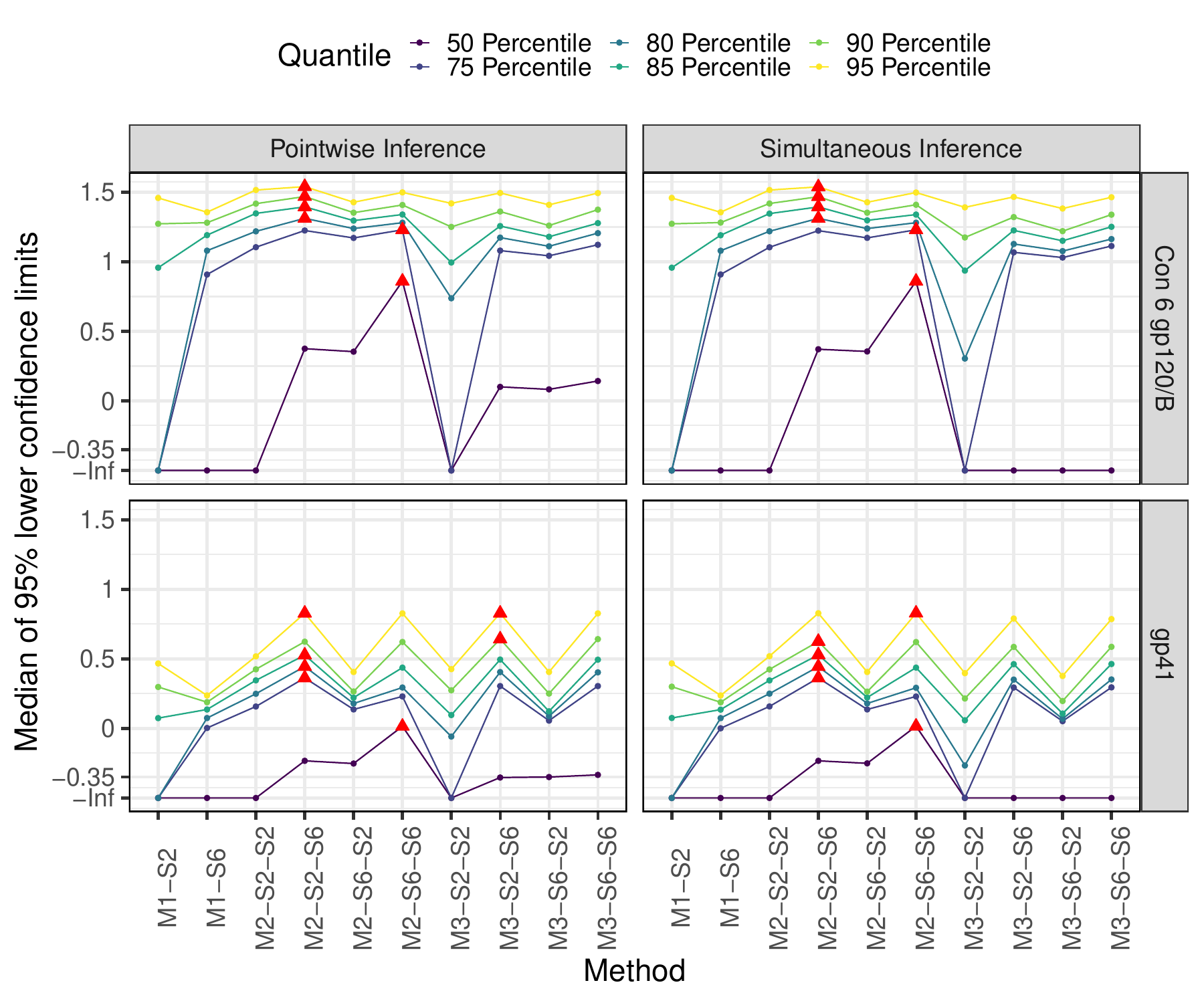}
    \caption{Median of the $95\%$ lower confidence limits of $\tau_{(k)}$ over 1000 simulations derived from each method, for \textbf{DGP I} (Con 6 gp 120/B) and \textbf{DGP II} (gp41) with $N_1 = N_0 = 30$ and $k = \lceil 0.5 \times N \rceil, \lceil 0.75 \times N \rceil, \lceil 0.8 \times N \rceil, \lceil 0.85 \times N \rceil, \lceil 0.9 \times N \rceil, \text{and }\lceil 0.95 \times N \rceil$. The red triangle represents the largest median for each quantile.}
    \label{fig: median 30}
\end{figure}

\begin{figure}[ht]
    \centering
    \includegraphics[width = 0.8\textwidth]{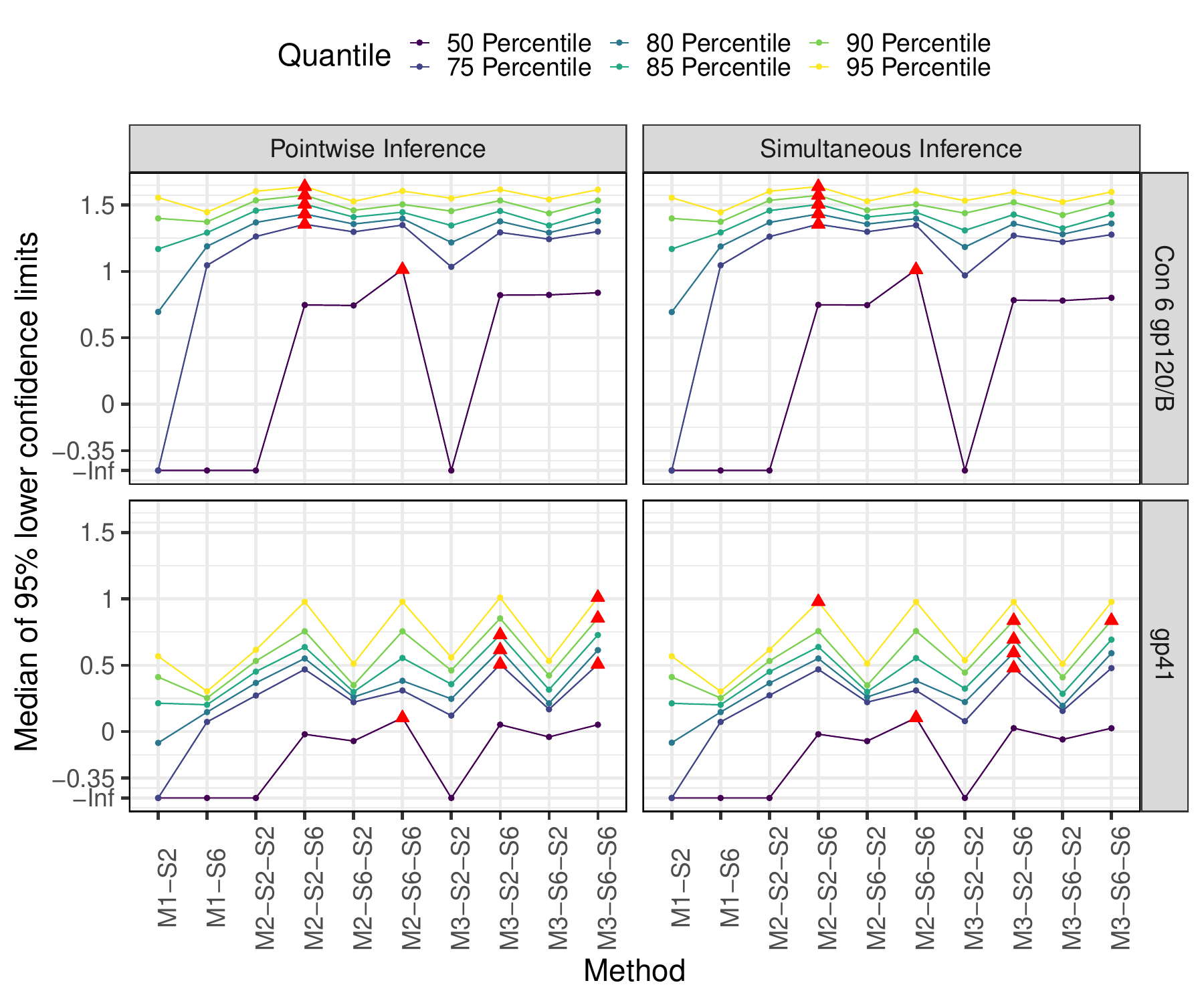}
    \caption{Median of the $95\%$ lower confidence limits of $\tau_{(k)}$ over 1000 simulations derived from each method, for \textbf{DGP I} and \textbf{DGP II} with $N_1 = N_0 = 100$ and $k = \lceil 0.5 \times N \rceil, \lceil 0.75 \times N \rceil, \lceil 0.8 \times N \rceil, \lceil 0.85 \times N \rceil, \lceil 0.9 \times N \rceil, \text{and }\lceil 0.95 \times N \rceil$. The red triangle represents the largest median for each quantile.}
    \label{fig: median 100}
\end{figure}

\section{Application to early phase HIV-1 vaccine trials}
\label{sec: case study}

\subsection{HVTN 086}
\label{subsec: HVTN 086}
Heterogeneity in vaccine-induced immune responses has been widely observed among vaccinees.
For instance, as shown in Figure \ref{fig: 086 heterogeneity}, vaccine-induced immune responses are heterogeneous among participants receiving 4 vaccine regimens in the HVTN 086 study, spanning from participants with no response (``nonresponders'') to those with exceptional responses (``high-responders"). Recently, \citet{huang2022baseline} conducted a comprehensive meta-analysis exploring the variations in immune responses induced by $26$ vaccine regimens. 
We focus on the immunogenicity data derived from HVTN 086, a multicenter, randomized, placebo-controlled phase-1 clinical trial studying 4 vaccine regimens: SAAVI MVA-C priming with sequential or concurrent Novartis subtype C
gp140/MF59 vaccine boost and SAAVI DNA-C2 priming with SAAVI MVA-C boosting, with or without Novartis subtype C gp140/MF59 vaccine.
The study comprised 4 study arms. Each arm planned to enroll $46$ HIV-negative, healthy, vaccine-na\"ive adult participants between $18$ and $45$ in the Republic of South Africa, among whom $38$ were randomly assigned to one candidate vaccine regimen and the other $8$ to the placebo. Study participants were randomized to one of the $4$ study arms, and within each study arm, further randomized to the active vaccine regimen or placebo. Below, we will follow \citet{huang2022baseline} and consider data from study participants who successfully completed all study visits and received all scheduled vaccination.

\subsection{Characterizing immunogenicity profiles against placebo}
\label{subsec: case study compare to placebo}
Our first goal is to characterize each vaccine regimen's immunogenicity profile of antigen-specific immune responses within each study arm.
The primary outcome of interest is the serum IgG response to the antigen Con6 gp120/B measured by a validated binding antibody multiplex assay (BAMA) $2$ weeks post the last vaccination.
 Because study participants were all na\"ive to the antigen, it is reasonable to assume that their potential outcomes under placebo were less than or equal to $100$, the limit of detection of the BAMA assay. 

Figure \ref{fig: realdata_compare_placebo} plots the $95\%$ one-sided confidence intervals for selected ITE quantiles using methods described in Section \ref{sec: treatment against placebo}. 
The vertical dashed red lines indicate the $95\%$ lower confidence limits of the sample average treatment effect derived from a randomization-based test based on the t-statistic \citep{neyman1923application, li2018asymptotic}. Inference for the effect quantiles largely enriches the data summary based on the SATE; for instance, in addition to stating that the lower confidence limit of SATE of the vaccine regimen \textsf{T1} (comparing to the placebo) is $2.15$ (in the $\log_{10}$-scale), researchers could further report at a specified confidence level that the largest ITE is at least $2.49$, the top $25\%$ ITE is at least $2.40$, and the median ITE is at least $2.30$, etc. Moreover, unlike the inference for the SATE that requires large-sample approximation, inference for the quantiles is \emph{exact}. Figure \ref{fig: realdata_compare_placebo} also suggests that the simultaneous confidence intervals are very similar to pointwise confidence intervals, suggesting that conducting simultaneous inference does not sacrifice much power. 

\begin{figure}[ht]
\centering
\subfloat[Inference for a single quantile of treatment effects]{
	\label{subfig: point_realdata_compare_placebo}
	 \includegraphics[width=0.47\textwidth]{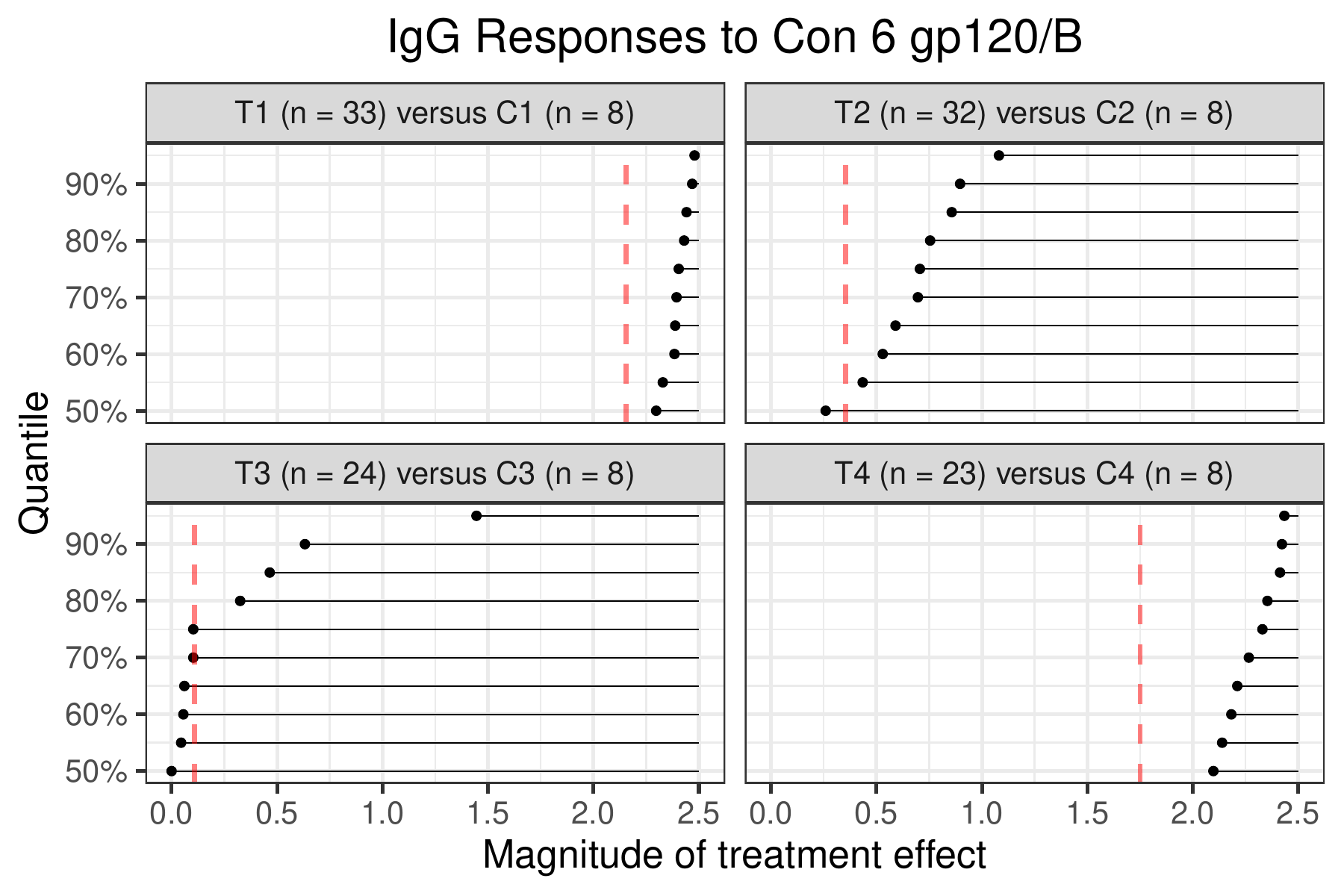}}
\subfloat[Simultaneous inference for multiple quantiles of treatment effects]{
	\label{subfig: simul_realdata_compare_placebo}
	\includegraphics[width=0.47\textwidth]{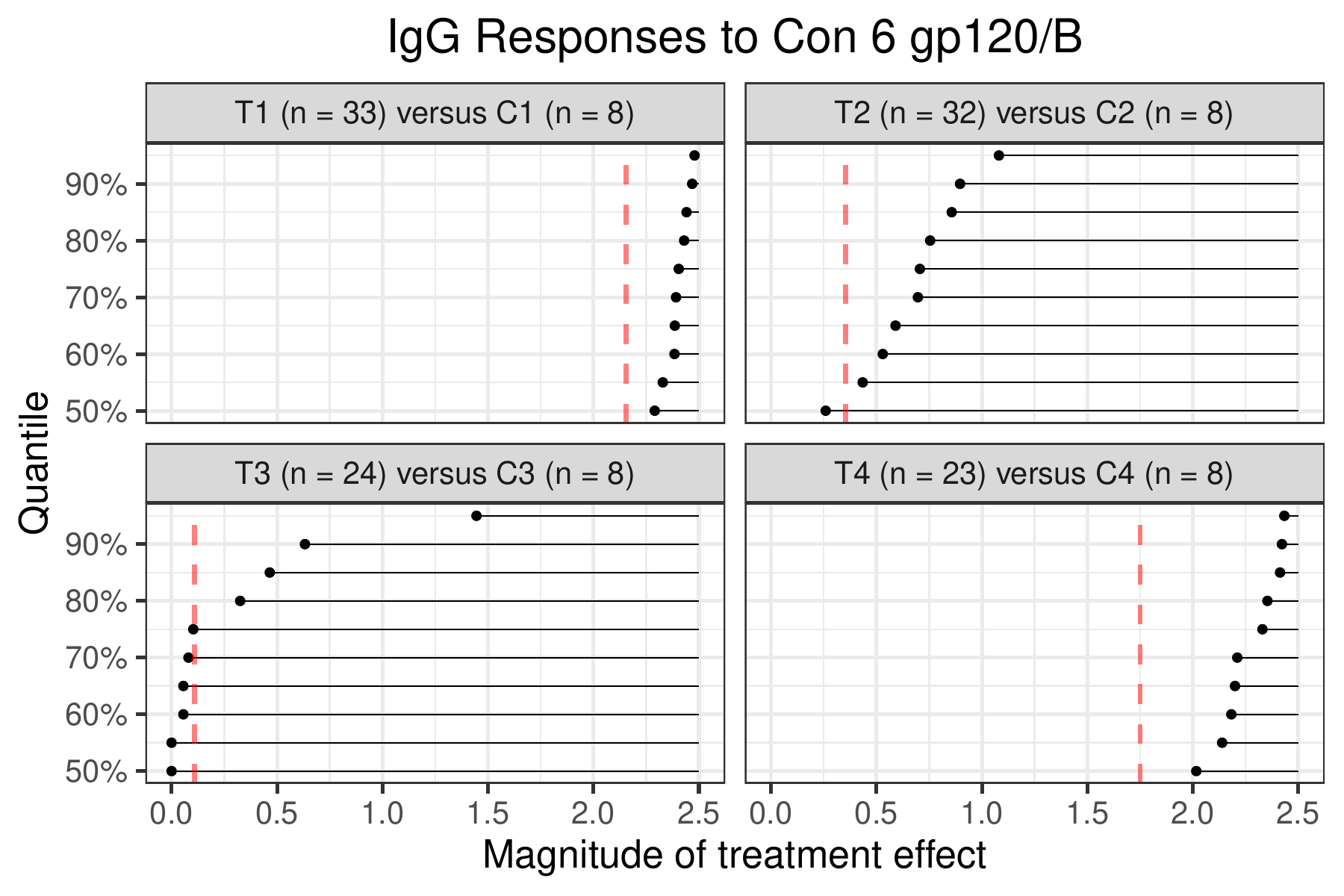}}
 \caption{Vaccine regimens versus placebo: $95\%$ (a) pointwise and (b) simultaneous one-sided confidence intervals of $\tau_{(k)}$'s for $k =  \lceil 0.5 \times N \rceil,  \lceil 0.55 \times N \rceil,  \lceil 0.6 \times N \rceil,  \lceil 0.65 \times N \rceil, \lceil 0.7 \times N \rceil,\lceil 0.75 \times N \rceil, \lceil 0.8 \times N \rceil, \lceil 0.85 \times N \rceil, \lceil 0.9 \times N \rceil, \text{and }\lceil 0.95 \times N \rceil$.}
\label{fig: realdata_compare_placebo}
\end{figure}

Inferred ITE profiles also facilitate comparing and ranking $4$ candidate regimens via making inference for $N(c)$, the number of participants with ITEs exceeding a given threshold $c$. Table \ref{table: realdata LB of N(c)} summarizes the $95\%$ lower confidence limits of $N(c)$ for each regimen and selected values of $c$. According to Table \ref{table: realdata LB of N(c)}, 
a $95\%$ confidence interval for $N(2)$ is $[31, 41]$ for \textsf{T1}, suggesting that we are $95\%$ confident that at least $31/41 = 75.6\%$ participants had a treatment effect as large as $2$ (in the $\log_{10}$-scale). Similarly, we are $95\%$ confident that at least $17/31 = 54.8\%$ participants had a treatment effect as large as $2$ (in the $\log_{10}$-scale) for regimen \textsf{T4}. Based on these results, if researchers are interested in advancing a vaccine regimen that can elicit high immune responses among a large proportion of participants, then \textsf{T1} is most promising based on data generated from this early phase clinical trial. 

\begin{table}[H]
\centering
\begin{tabular}{ccccccc}
  Regimen / c & 0 & 0.5 & 1 & 1.5 & 2 \\ 
  \hline
 T1 (n = 33) vs C1 (n = 8) &  40 &  40 &  40 &  40 &  31 \\ 
   T2 (n = 32) vs C2 (n = 8) &  27 &  17 &   3 &   1 &   0 \\ 
   T3 (n = 24) vs C3 (n = 8) &  15 &   4 &   3 &   0 &   0 \\ 
   T4 (n = 23) vs C4 (n = 8) &  29 &  29 &  24 &  23 &  17 \\ 
\end{tabular}
\caption{Pointwise $95\%$ lower confidence limits of $N(c)$ for four vaccine regimens, with $c=0, 0.5, 1, 1.5$ and $2$.}
\label{table: realdata LB of N(c)}
\end{table}

\subsection{Head-to-head comparisons of two vaccine regimens}
\label{subsec: compare two vaccine regimen}
Having a placebo arm and an \emph{a priori} known $Y(0)$ is a favorable scenario. Below, we consider a direct, head-to-head comparison of three pairs of vaccine regimens in HVTN 086: \textsf{T1} vs. \textsf{T2}, \textsf{T1} vs. \textsf{T3}, and \textsf{T1} vs. \textsf{T4}. A head-to-head comparison of two active vaccine regimens is useful because in most recent early phase vaccine trials, a placebo arm is no longer employed and study participants are often randomized to different active vaccine regimens. Figure \ref{fig: head_to_head_comparison} shows the $95\%$ one-sided simultaneous confidence intervals for selected ITE quantiles of regimen \textsf{T1} versus the other three regimens.
The vertical dashed red
lines show the $95\%$ lower confidence limits of the SATE derived from a randomization test based on the t-statistic. For both Con 6 gp120/B and gp41, we conducted ITE inference with the method \textsf{M2}. Specifically, we chose the Stephenson rank sum statistic with $s = 6$ in both Step (A1) and Step (A2) for Con 6 gp120/B; as for gp41, we again used the Stephenson rank sum statistic but with $s=2$ in Step (A1) and $s = 6$ in Step (A2). These choices were guided by the simulation studies in Section \ref{sec: simulation}. 


\begin{figure}[ht]
\centering
\subfloat[Simultaneous inference for Con 6 gp120/B using \textsf{M2-S6-S6}]{
	\label{subfig: con6_head_to_head}
	 \includegraphics[width=0.47\textwidth]{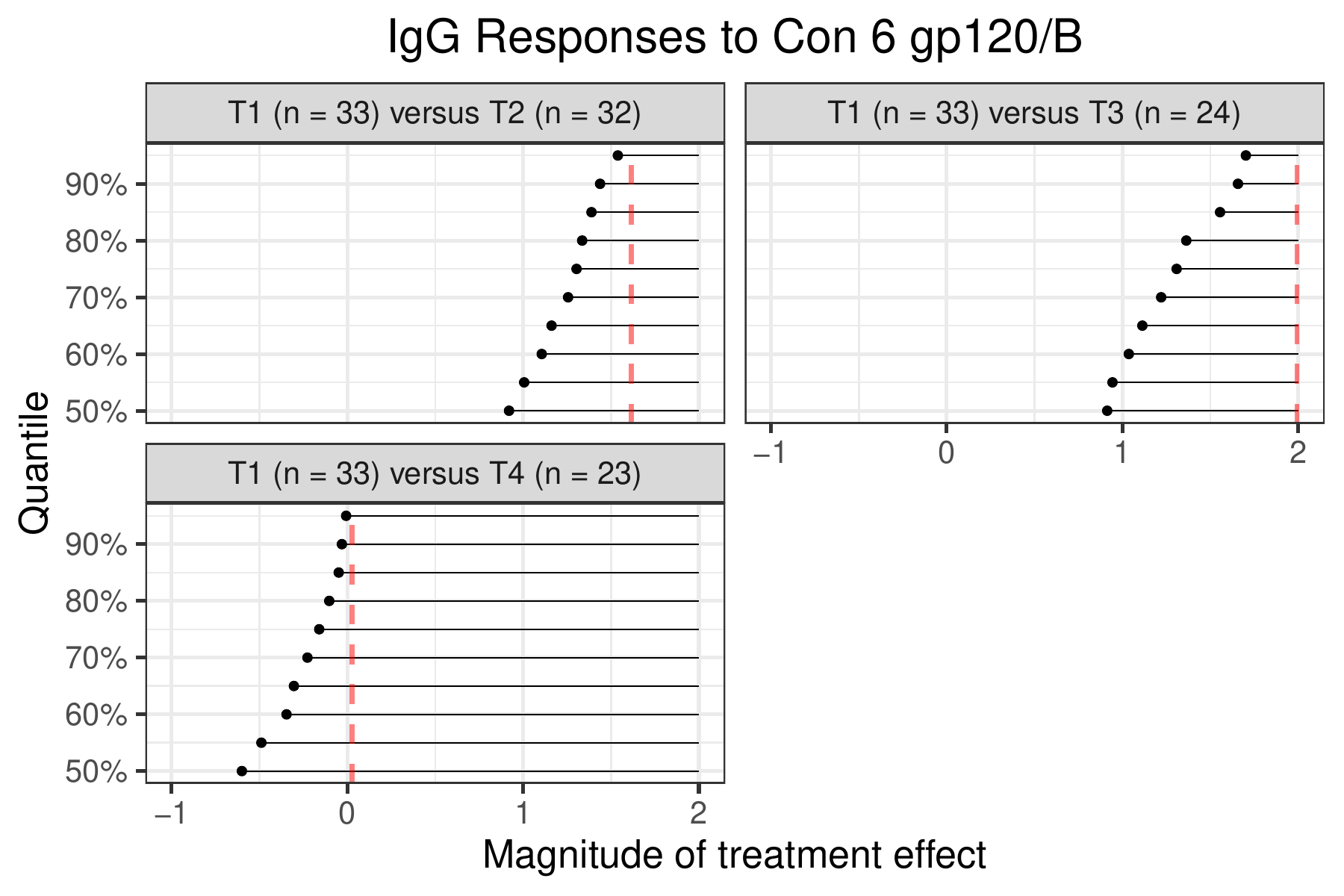}}
\subfloat[Simultaneous inference for gp41 using \textsf{M2-S2-S6}]{
	\label{subfig: gp41_head_to_head}
	\includegraphics[width=0.47\textwidth]{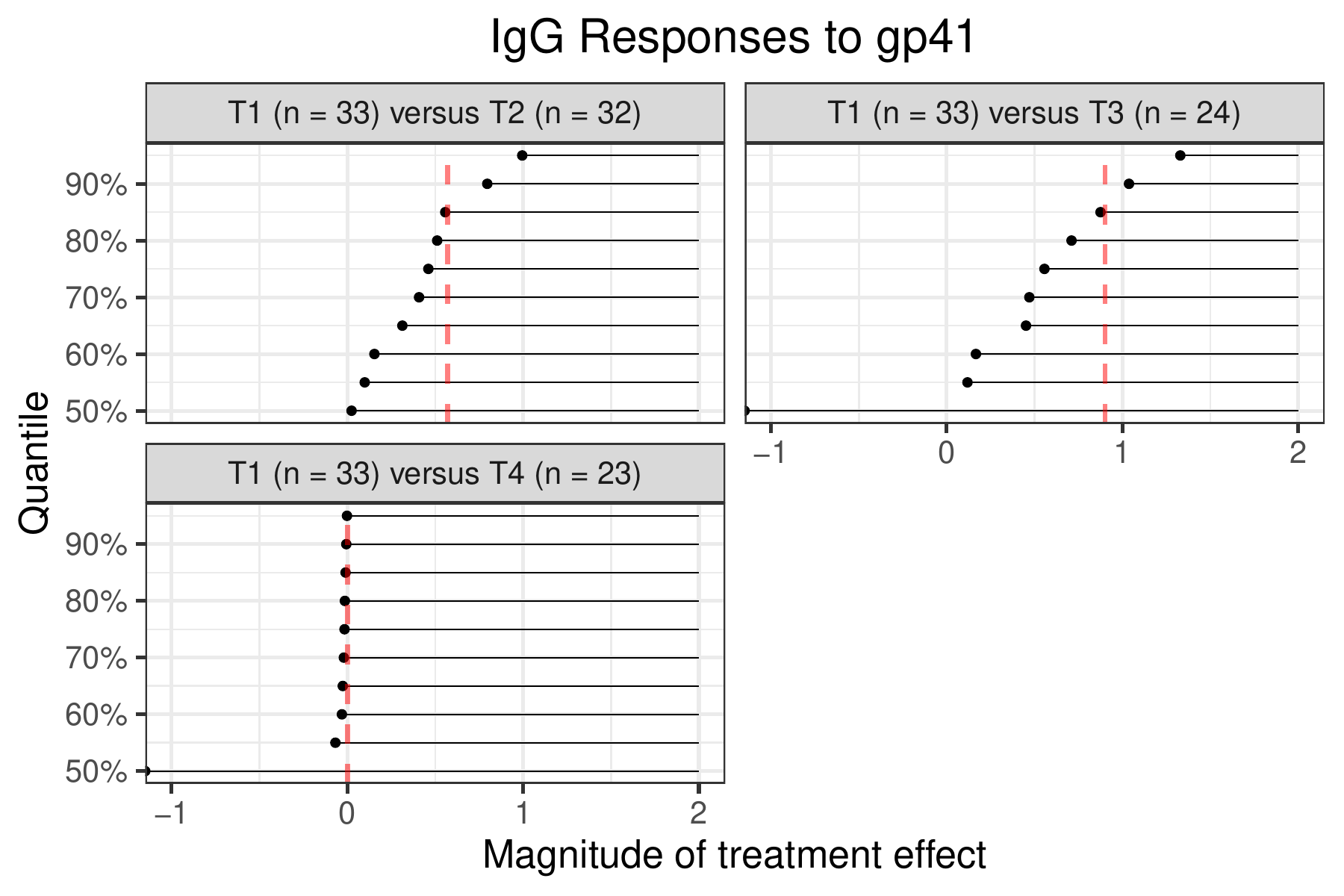}}
 \caption{Head-to-head comparison of vaccine regimens in terms of the binding antibody response to (a) Con 6 gp120/B  and (b) gp41: $95\%$ simultaneous one-sided confidence intervals of $\tau_{(k)}$'s for $k =  \lceil 0.5 \times N \rceil,  \lceil 0.55 \times N \rceil,  \lceil 0.6 \times N \rceil,  \lceil 0.65 \times N \rceil, \lceil 0.7 \times N \rceil,\lceil 0.75 \times N \rceil, \lceil 0.8 \times N \rceil, \lceil 0.85 \times N \rceil, \lceil 0.9 \times N \rceil, \text{and }\lceil 0.95 \times N \rceil$.}
\label{fig: head_to_head_comparison}
\end{figure}



According to Figure \ref{fig: head_to_head_comparison}(b), the $95\%$ lower confidence limit for the SATE comparing \textsf{T1} versus \textsf{T2} is $0.74$ for gp41. Inference for ITE further reveals treatment effect heterogeneity of \textsf{T1} versus \textsf{T2}: we are $95\%$ confident that more than 
$12.3\%$ of participants benefited more from \textsf{T1} compared to \textsf{T2} by at least $0.74$ in the $\log_{10}$-scale (or equivalently a $5.50$ times increase in the raw readout) and, at the same time, at least $4.6\%$  participants benefited by more than $1$ unit in the $\log_{10}$-scale (or equivalently $10$ times) and $1.5\%$ benefited by $1.25$ (or equivalently $17.78$ times). In fact, because the inference for ITE is simultaneously valid, we can immediately reject a constant additive treatment effect of $0.74$ for all participants at $95\%$ confidence level. Analogous inference can be made when comparing \textsf{T1} to \textsf{T3}.



On the other hand, according to Figure \ref{fig: head_to_head_comparison}(a), the confidence intervals for ITE quantiles comparing \textsf{T1} versus \textsf{T2} for Con 6 gp120/B cover $1.61$, the lower confidence limit of the SATE, simultaneously. For Con 6 gp120/B, a constant additive treatment effect of $1.61$ for all participants is in fact compatible with the observed data and cannot be rejected at 95\% confidence interval. Compared to gp41, response to Con 6 gp120/B appears to be less heterogeneous when comparing \textsf{T1} versus \textsf{T2} or \textsf{T3}. These conclusions are consistent with the visual display of the data in Figure \ref{fig: 086 heterogeneity}.



\subsection{A matched study of non-randomized arms}
\label{subsec: case study across trial comparison Gamma = 1}

We next considered a head-to-head, across-trial comparison of two regimens --- regimen \textsf{T1} from HVTN 086 (n = 33) and regimen \textsf{T4} from HVTN 205 (n = 60) --- based on the serum IgG response to
antigen gp41 \citep{huang2022baseline}. HVTN 205 was a phase 2a study designed to evaluate DNA and recombinant modified vaccinia Ankara (MVA62B) vaccines \citep{goepfert2014specificity}. To alleviate the ``trial selection bias," we used statistical matching to control for observed baseline characteristics. Specifically, we used an optimal tripartite matching method \citep{zhang2023matching} and constructed $33$ matched pairs, each consisting of one study participant receiving vaccine regimen \textsf{T1} from HVTN 086 and the other receiving \textsf{T4} from HVTN 205. The matching algorithm closely matched on $11$ baseline covariates and minimized the earth mover's distance between the estimated propensity score distributions of two groups. Table \ref{table: cov_balance} summarizes the covariate balance before and after matching; after matching, two groups are more balanced in baseline characteristics, with the standardized mean differences of most variables below $0.20$, or one-fifth of one pooled standard deviation, which is typically considered good covariate balance \citep{silber2001matching}. Figure \ref{fig: case study across arm Gamma = 1}(a) plots the response magnitude among matched study participants. Both regimens elicited strong binding antibody response to gp41 among participants.
 
\begin{table}[ht]
\centering
 {\begin{tabular}{lccccc} 
        \multirow{3}{*}{Covariates} & \multicolumn{3}{c}{\textbf{Before matching}}& \multicolumn{2}{c}{\textbf{After matching}}\\
          & \text{HVTN 086}& \text{HVTN 205}& SMD & \text{HVTN 205}& SMD \\
          & \text{T1 (n = 33)}& \text{T4 (n = 60)}&  & \text{T4 (n = 33)}& \\
  \hline
Age (years) &  23.6 & 26.5 & -0.37 & 23.9 & -0.04 \\
  Sex assigned at birth (female or male) &  0.45 & 0.67 & -0.31 & 0.58 & -0.17  \\ 
  BMI (kg/$m^2$) & 23.4 & 25.0 & -0.24 & 24.0 & -0.09  \\ 
  Systolic blood pressure (mm Hg) & 119.0 & 114.6 & 0.33 & 116.7 & 0.17  \\ 
  Diastolic blood pressure (mm Hg) & 74.3 & 71.2 & 0.29 & 71.9 & 0.22 \\ 
  Hematocrit (\%) & 42.5 & 42.6 & -0.03 & 42.7 & -0.04  \\ 
  Hemoglobin (/$\mu$L) & 14.1 & 14.5 & -0.19 & 14.4 & -0.13 \\ 
  Lymphocytes (/$\mu$L) & 1944 & 1986 & -0.05 & 1924 & 0.03\\ 
  Mean corpuscular volume (fL/red cell) & 88.2 & 89.2 & -0.14 & 89.1 & -0.13 \\ 
  Neutrophils (/$\mu$L) & 3501 & 3569 & -0.04 & 3411 & 0.05 \\ 
  Platelets (/nL) & 270.0 & 244.9 & 0.28 & 247.9 & 0.25 \\ 
\end{tabular}}
\caption{Covariate balance before and after matching. SMD = standardized mean difference.}
\label{table: cov_balance}
\end{table}

We first conducted inference under a randomization assumption: two study participants in each matched pair have the same probability of receiving HVTN 086 \textsf{T1} or HVTN 205 \textsf{T4}; in other words, the matched observational study succeeded in embedding data into a finely stratified randomized experiment \citep{rosenbaum2002observational,chen2023testing}.  We conducted randomization inference for ITE quantiles of HVTN 086 \textsf{T1} versus HVTN 205 \textsf{T4} using \citeauthor{CL}'s method \citeyearpar{CL} reviewed in Section \ref{subsec: stratified experiment} with the stratified Wilcoxon rank sum statistic \citep{10.1093/biomet/asad030}.
Figure \ref{fig: case study across arm Gamma = 1}(b) shows the simultaneous one-sided confidence intervals for selected ITEs $\tau_{(k)}$'s. According to Figure \ref{fig: case study across arm Gamma = 1}(b), the $95\%$ lower confidence limit for ITE at rank
$51$ barely exceeds $0$. This implies that, at $95\%$ confidence level, at least $24\%$, or equivalently $16$ out of $66$, study participants benefited more from HVTN 086 \textsf{T1} compared to HVTN 205 \textsf{T4} in this matched study. We further complemented the analysis by making randomization-based inference for SATE in a matched-pair design \citep{imai2008variance}. The $95\%$ lower confidence limit of SATE is $0.086$ (in the
$\log_{10}$-scale) and quite sensitive due to a few outliers in the HVTN 205 \textsf{T4} arm. The exact inference for the proportion of study participants who benefited more from HVTN 086 \textsf{T1} versus HVTN 205 \textsf{T4} helps complement a standard analysis of the sample average treatment effect.


\begin{figure}%
    \centering
    \subfloat[\centering Distributions of observed immune marker measurements]
    {{\includegraphics[width=0.45\textwidth]{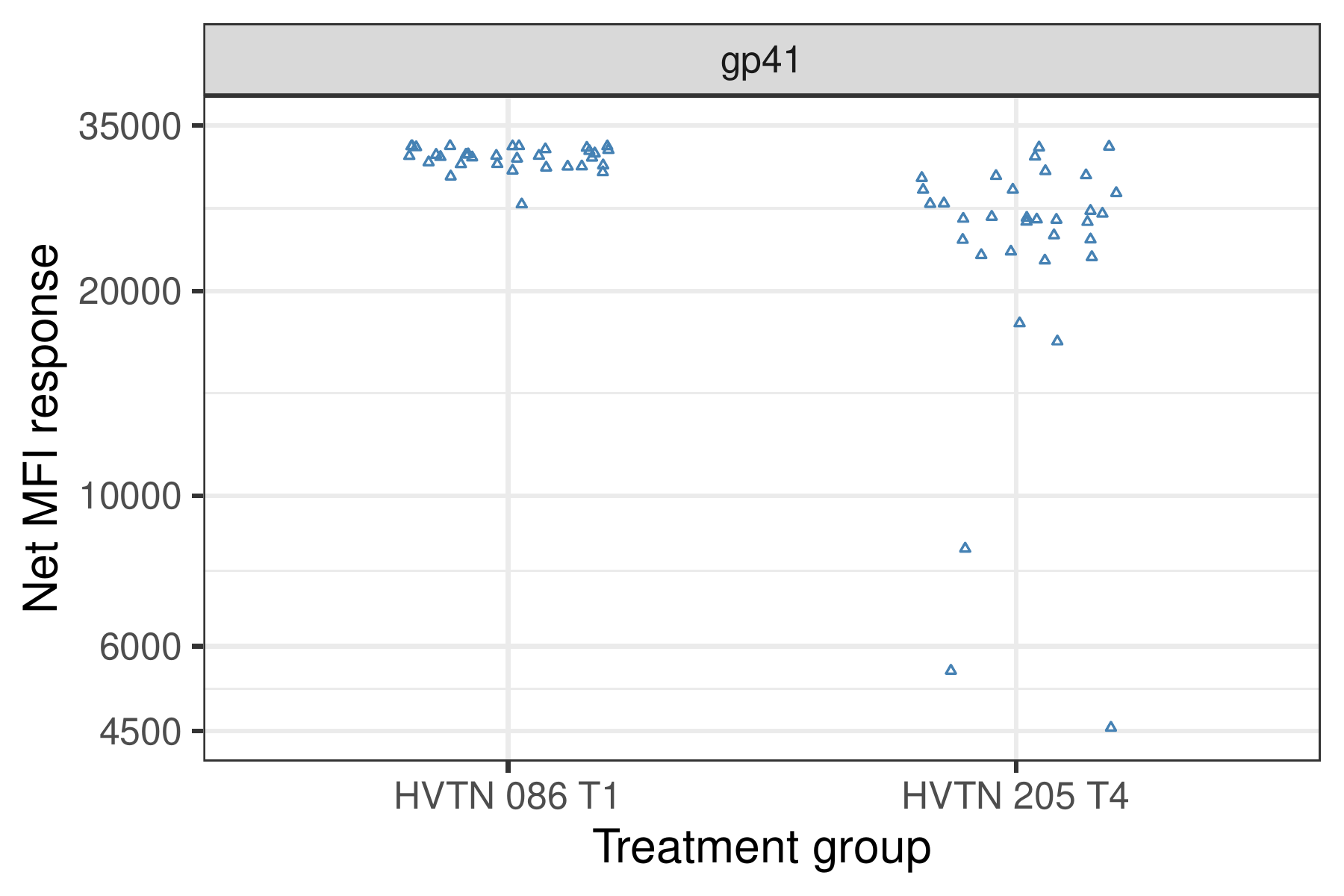} }}%
    \qquad
    \subfloat[\centering Randomization inference for ITE quantiles]{{\includegraphics[width=0.45\textwidth]{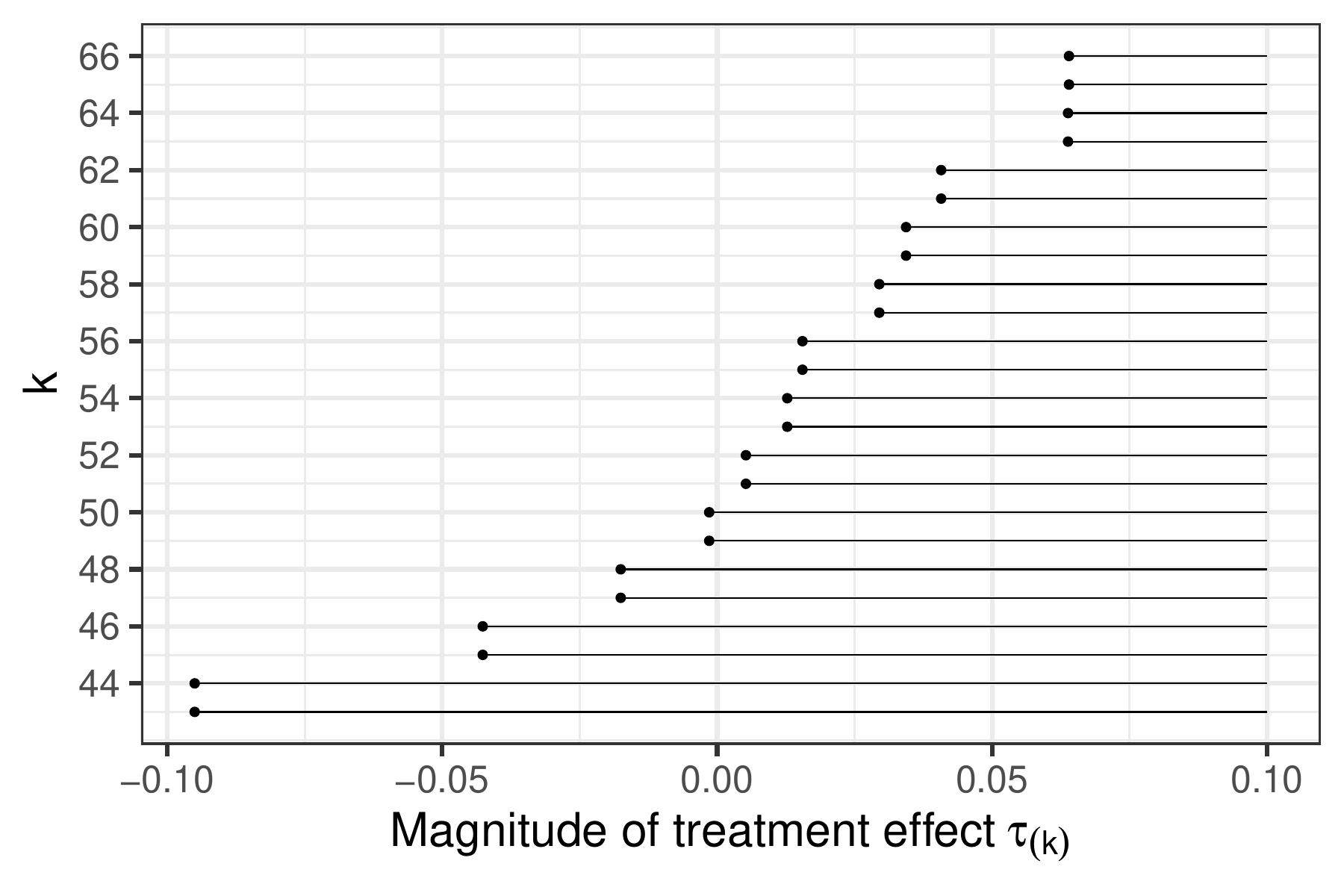} }}%
    \caption{(a) Observed serum IgG BAMA responses to gp41  among study participants who received regimen T1 in HVTN 086 and those who received regimen T4 in HVTN 205.
(b) $95\%$ simultaneous one-sided confidence intervals for ITE quantiles of HVTN 086 T1 versus HVTN 205 T4 using the stratified Wilcoxon rank sum statistic, assuming that there is no hidden confounding.
    }%
    \label{fig: case study across arm Gamma = 1}
\end{figure}

\subsection{A sensitivity analysis assessing deviation from randomization}
\label{subsec: case study SA}
As is true for all non-randomized studies, the head-to-head comparison between HVTN 086 \textsf{T1} and HVTN 205 \textsf{T4} in Section \ref{subsec: case study across trial comparison Gamma = 1} could suffer from unmeasured confounding bias. For instance, this would be the case if HVTN 205 enrolled healthier and more immunopotent participants compared to HVTN 086; hence, any treatment effect comparing HVTN 086 \textsf{T1} to HVTN 205 \textsf{T4} could be attributed to this hidden bias \citep{rosenbaum2002observational}. We next conduct a sensitivity analysis that investigates how a deviation from the randomization assumption could impact the inferred ITE quantiles in Figure \ref{fig: case study across arm Gamma = 1}(b). 

Figure \ref{fig: SA} shows the sensitivity analysis results derived from \citeauthor{CL}'s method \citeyearpar{CL} using the stratified Wilcoxon rank sum statistic. Figure \ref{fig: SA} plots the $95\%$ confidence intervals for ITE quantiles under Rosenbaum's sensitivity analysis model indexed by $\Gamma$; see Equation \eqref{eq: Gamma model} \citep{rosenbaum2002observational}.
From Figure \ref{fig: SA}, when the bias in the treatment assignment is as large as $\Gamma = 1.5$, the ITE at rank $57$, i.e., approximately $86\%$ quantile, remains positive at significance level $0.05$, which implies that HVTN 086 \textsf{T1} would still induce a higher MFI response than HVTN 205 \textsf{T4} for at least $15\%$ of study participants in the study under a moderate bias of magnitude $1.5$. 
According to the method described in \citet{rosenbaum2009amplification}, an unobserved covariate associated with at least a $2.5$--fold increase in the odds of selecting in the study arm \textsf{T1/C1} (HVTN 086) as opposed to \textsf{T4/C4} (HVTN 205) \emph{and} a $2.75$--fold increase in the odds of a positive matched-pair difference in MFI values is needed in order to explain away the top $15\%$ ITE. 
Such a confounding factor appears unlikely after we have controlled for observed covariates using matching. Therefore, we may conclude that HVTN 086 \textsf{T1} induces higher MFI response compared to HVTN 205 \textsf{T4} for at least $15\%$ of the cohort even in the presence of a moderately large selection. When $\Gamma$ increases to $3.3$, the resulting $95\%$ confidence intervals no longer cover zero for any participant, implying that we do not have evidence that any participant would benefit more from HVTN 086 \textsf{T1} versus HVTN 205 \textsf{T4} if the trial selection bias is as large as $\Gamma = 3.3$.

\begin{figure}
    \centering
   \subfloat[$\Gamma = 1.2$]{
	 \includegraphics[width=0.47\textwidth]{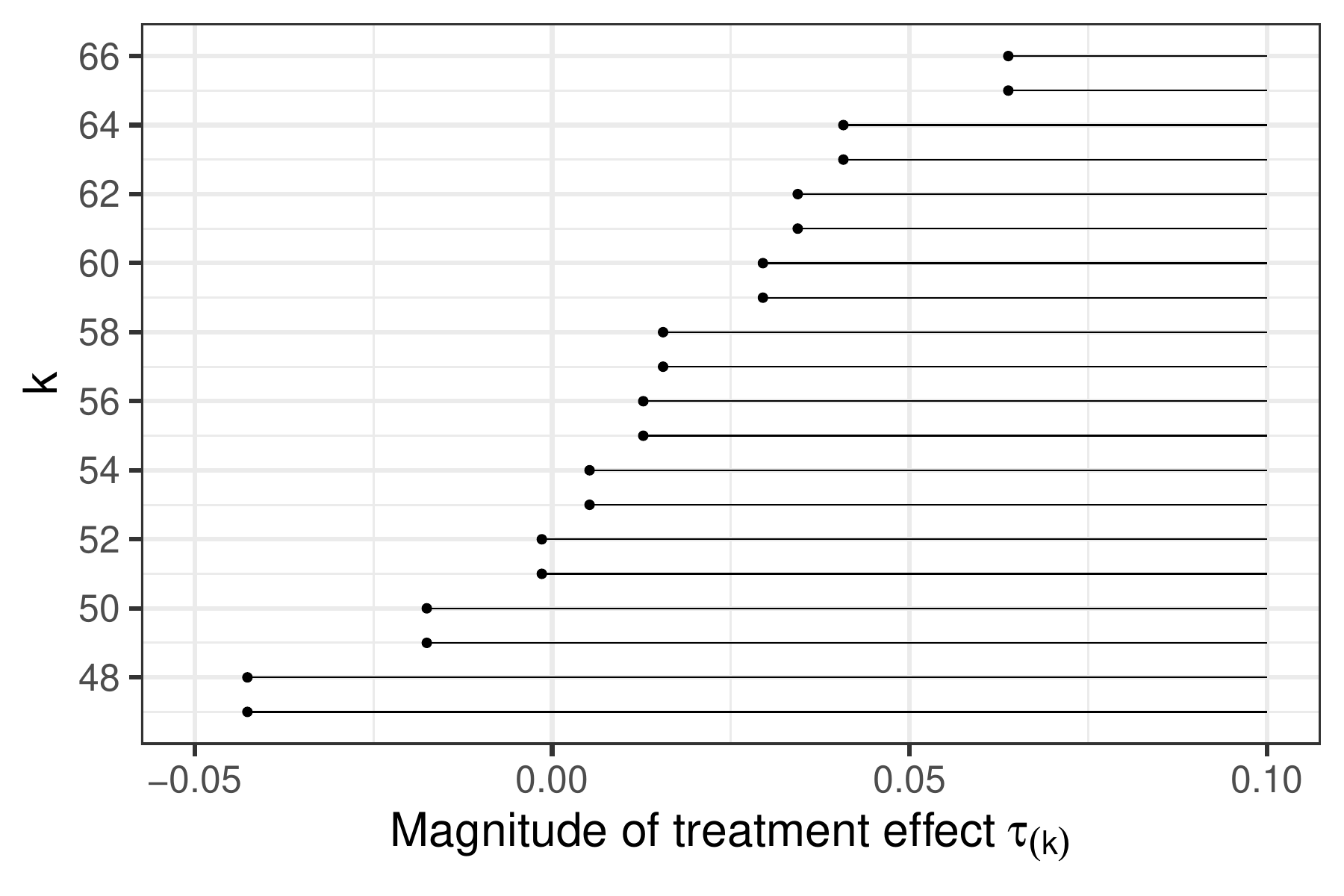}} \hfill
  \subfloat[$\Gamma = 1.5$]{
	 \includegraphics[width=0.47\textwidth]{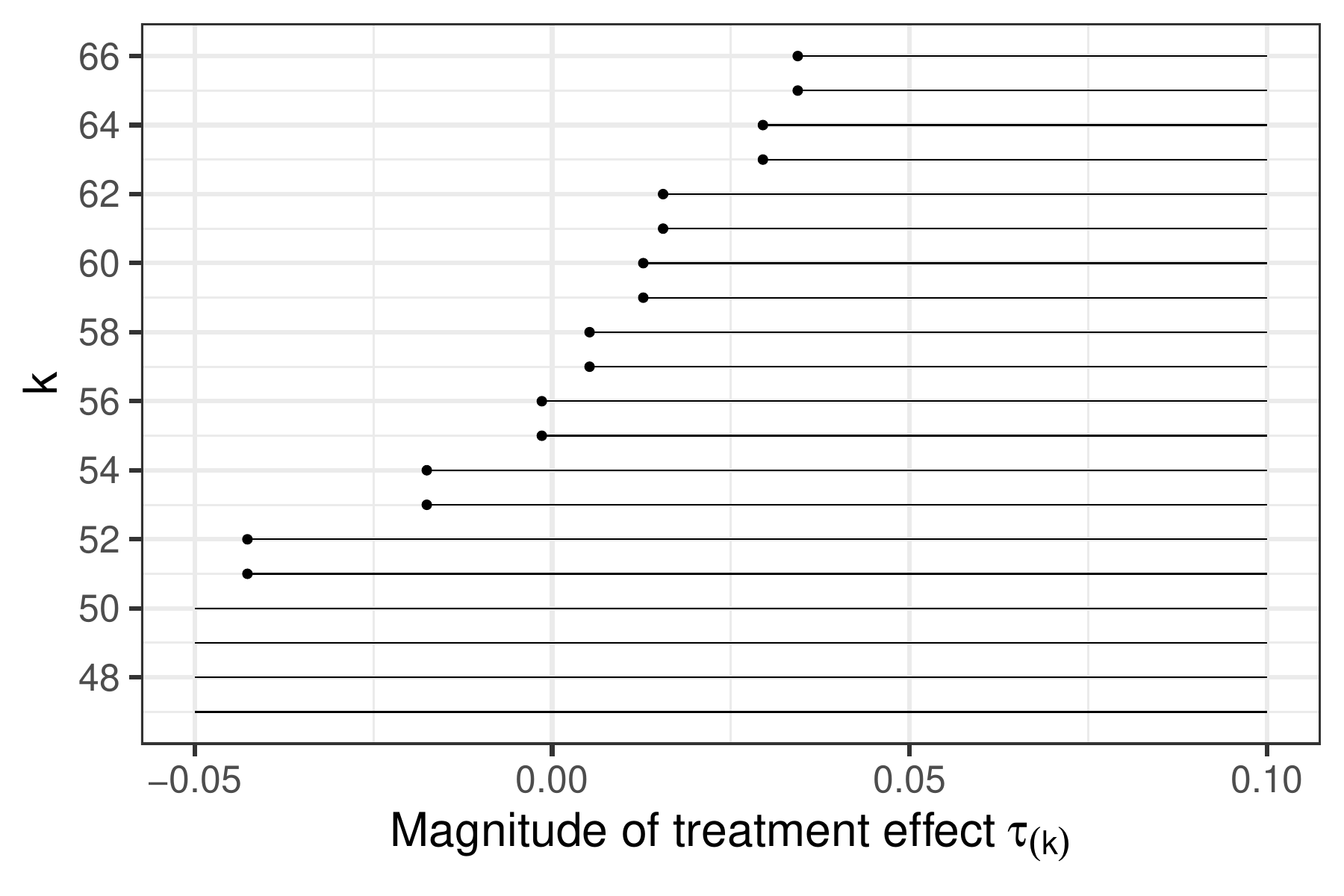}} \\
  \subfloat[$\Gamma = 2.5$]{
	 \includegraphics[width=0.47\textwidth]{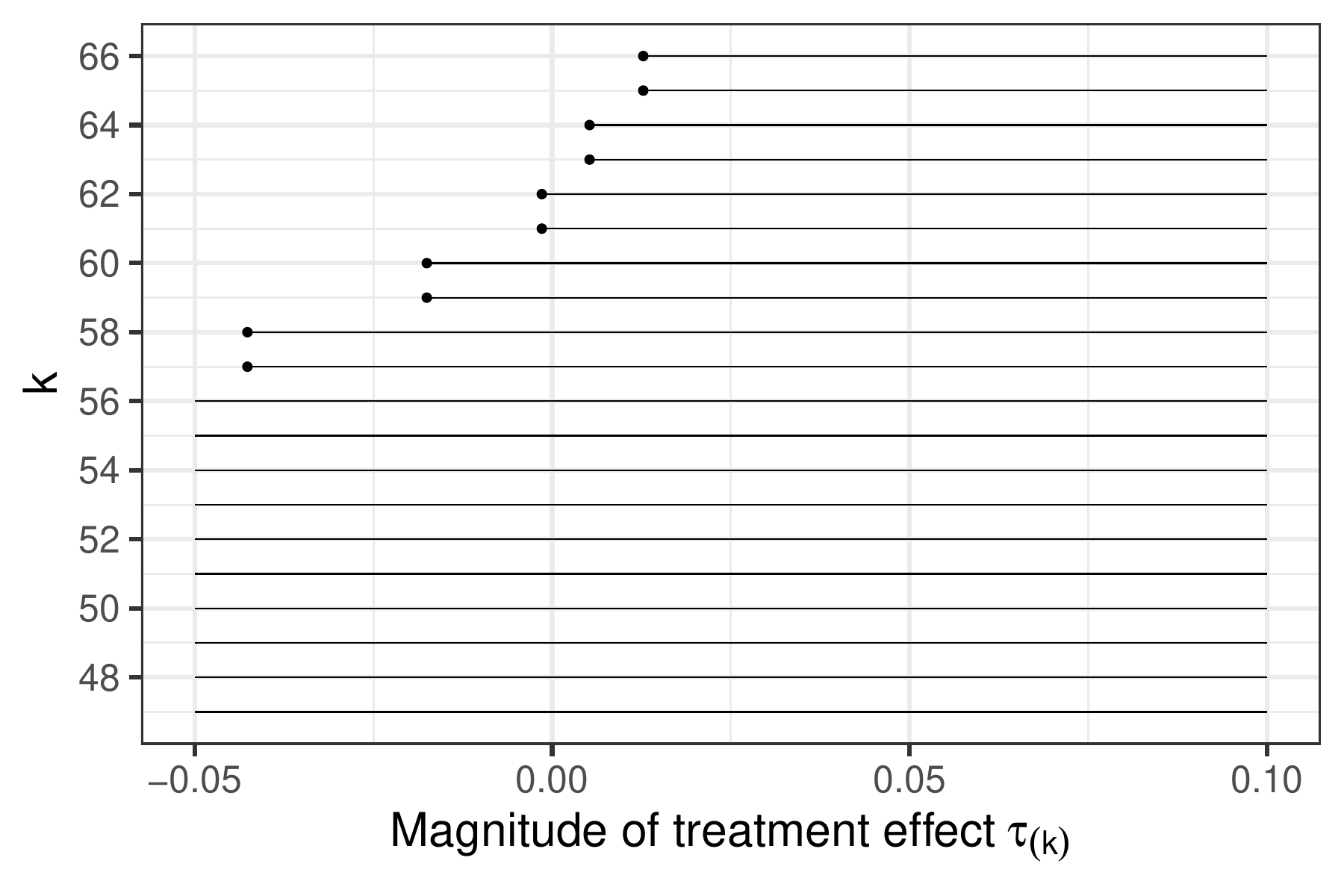}} \hfill
  \subfloat[$\Gamma = 3.3$]{
	 \includegraphics[width=0.47\textwidth]{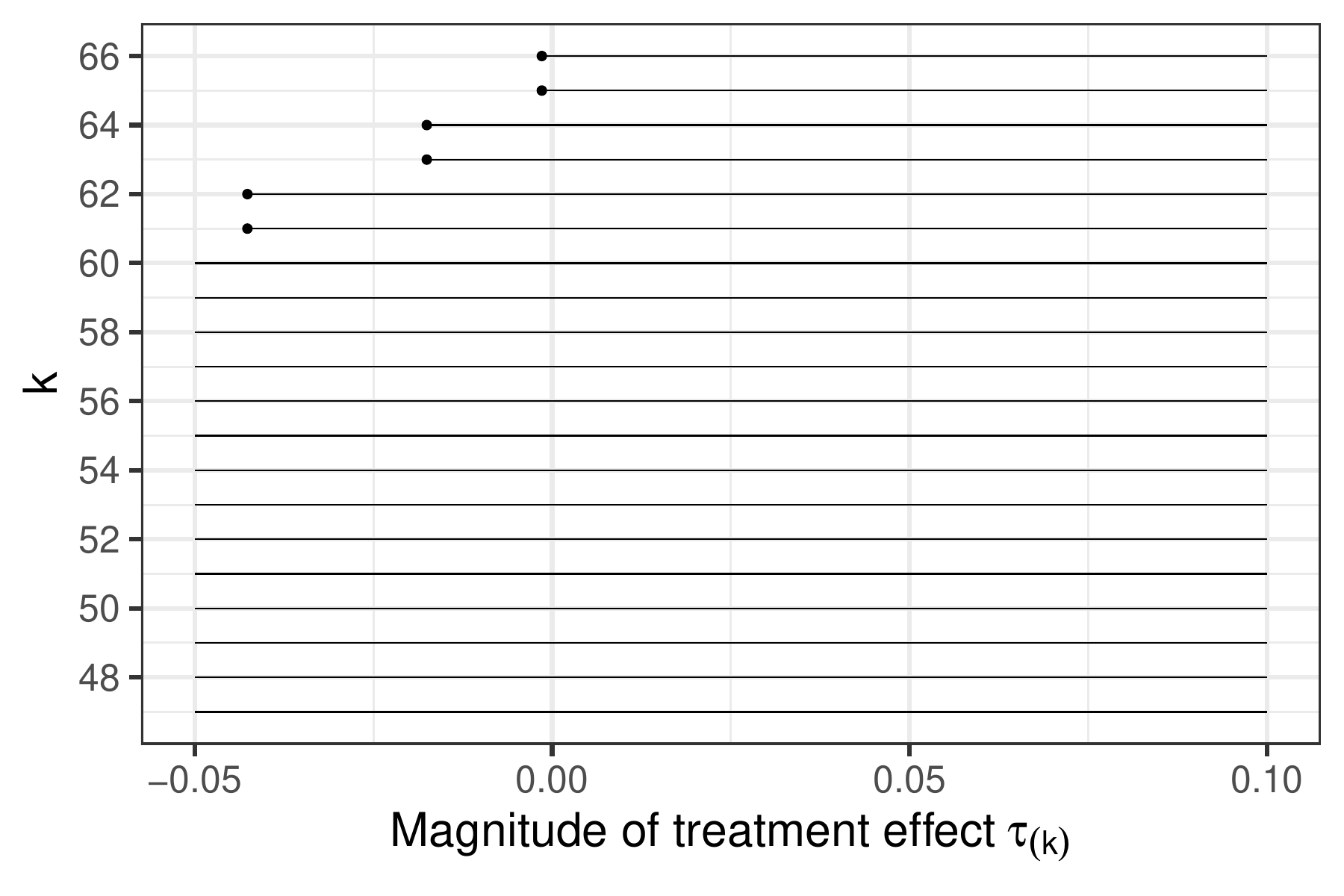}}
  \caption[ The average and standard deviation of critical parameters ]
        {\small  
        The $95\%$ simultaneous one-sided confidence intervals for ITE quantiles under various sensitivity models indexed by different $\Gamma$ using the enhanced method in \cite{CL}.} 
        \label{fig: SA}
\end{figure}

\section{Discussion}
\label{sec: discussion}
Many recent research has centered around moving beyond Fisher's sharp null hypothesis and exploring aspects of the collection of individual treatment effects other than the sample mean. In this article, we provide a systematic review of relevant methods and closely examine the usefulness of different methods under a wide range of scenarios. We found that randomization inference that tests quantile treatment effects could be a useful complement to the SATE, especially when the scientific interest lies in uncovering and quantifying the heterogeneous treatment effects. In addition, these methods hold promise in helping relevant stakeholders advance an experimental therapy that has a meaningfully large treatment effect on possibly a fraction of study participants, as opposed to a competing therapy that has a similar SATE but nevertheless does not show any treatment effect at a magnitude of practical relevance. Another interesting and practically relevant finding is that, perhaps contrary to expectation, when constructing a simultaneous confidence region jointly for many quantiles, or even every quantile, the cost of multiple hypothesis testing could be minimal.

Overall, our assessment is that the suite of randomization-based inferential methods, testing Fisher's sharp null of no effect whatsoever, Neyman's weak null of no SATE, and various quantiles of ITEs discussed in the current article, should always have a place in empirical researchers' toolbox when analyzing data derived from clinical trials because these methods are always reliable and could potentially yield informative results. 

In many senses, quantiles of ITEs are the most fine-grained estimands. This important line of research could be furthered in at least two directions. First, choosing an appropriate test statistic is a key component. As demonstrated in simulation studies, the power of the Stephenson rank sum statistic depends critically on the choice of $s$ and an optimal choice of the test statistic could depend largely on the data-generating process. One possibility is to develop a data-driven, adaptive approach that produces an optimal or near-optimal choice of the test statistic (or the tuning parameter in the test statistic) or combines several test statistics.
Second, we have reviewed and considered methods that (i)  assume a constant $Y(0)$; (ii) do not place any restriction on any aspect of the potential outcomes; (iii) place a uniform upper or lower bound on  potential outcomes. There are still plenty of other possibilities to leverage auxiliary information from historical data and improve the power of statistical inference.

\section*{Acknowledgements}
Research reported in this publication was partially supported by the National Science Foundation (award DMS-2238128 to Xinran Li). 

\section*{Data Availability Statement}
The data that support the findings of this study are available from the public-facing HIV Vaccine Trials Network (HVTN) website (\url{https://atlas.scharp.org/}).

\bibliographystyle{apalike}
\bibliography{paper-ref}

\begin{thebibliography}{}

\bibitem[Berger and Boos, 1994]{Berger1994}
Berger, R.~L. and Boos, D.~D. (1994).
\newblock P values maximized over a confidence set for the nuisance parameter.
\newblock {\em Journal of the American Statistical Association}, 89:1012--1016.

\bibitem[Caughey et~al., 2023]{caughey2021randomization}
Caughey, D., Dafoe, A., Li, X., and Miratrix, L. (2023).
\newblock Randomization inference beyond the sharp null: Bounded null hypotheses and quantiles of individual treatment effects.
\newblock {\em Journal of the Royal Statistical Society, Series B (Statistical Methodology)}, in press.

\bibitem[Chen et~al., 2023]{chen2023testing}
Chen, K., Heng, S., Long, Q., and Zhang, B. (2023).
\newblock Testing biased randomization assumptions and quantifying imperfect matching and residual confounding in matched observational studies.
\newblock {\em Journal of Computational and Graphical Statistics}, 32(2):528--538.

\bibitem[Chen and Li, 2023]{CL}
Chen, Z. and Li, X. (2023).
\newblock Distributions of individual treatment effects in sampling-based randomized experiments.
\newblock Technical Report.

\bibitem[Cohen and Fogarty, 2022]{cohen2022gaussian}
Cohen, P.~L. and Fogarty, C.~B. (2022).
\newblock Gaussian prepivoting for finite population causal inference.
\newblock {\em Journal of the Royal Statistical Society Series B: Statistical Methodology}, 84(2):295--320.

\bibitem[Dahabreh et~al., 2019]{dahabreh2019generalizing}
Dahabreh, I.~J., Robertson, S.~E., Tchetgen, E.~J., Stuart, E.~A., and Hern{\'a}n, M.~A. (2019).
\newblock Generalizing causal inferences from individuals in randomized trials to all trial-eligible individuals.
\newblock {\em Biometrics}, 75(2):685--694.

\bibitem[Ding and Dasgupta, 2018]{ding2018randomization}
Ding, P. and Dasgupta, T. (2018).
\newblock A randomization-based perspective on analysis of variance: a test statistic robust to treatment effect heterogeneity.
\newblock {\em Biometrika}, 105(1):45--56.

\bibitem[Ditse et~al., 2020]{ditse2020effect}
Ditse, Z., Mkhize, N.~N., Yin, M., Keefer, M., Montefiori, D.~C., Tomaras, G.~D., Churchyard, G., Mayer, K.~H., Karuna, S., Morgan, C., et~al. (2020).
\newblock Effect of hiv envelope vaccination on the subsequent antibody response to hiv infection.
\newblock {\em Msphere}, 5(1):e00738--19.

\bibitem[Fisher, 1935]{Fisher1935design}
Fisher, R.~A. (1935).
\newblock {\em The Design of Experiments.}
\newblock Oliver and Boyd. London and Edinburgh.

\bibitem[Fogarty, 2020]{fogarty2020studentized}
Fogarty, C.~B. (2020).
\newblock Studentized sensitivity analysis for the sample average treatment effect in paired observational studies.
\newblock {\em Journal of the American Statistical Association}, 115(531):1518--1530.

\bibitem[Goepfert et~al., 2014]{goepfert2014specificity}
Goepfert, P.~A., Elizaga, M.~L., Seaton, K., Tomaras, G.~D., Montefiori, D.~C., Sato, A., Hural, J., DeRosa, S.~C., Kalams, S.~A., McElrath, M.~J., et~al. (2014).
\newblock Specificity and 6-month durability of immune responses induced by dna and recombinant modified vaccinia ankara vaccines expressing hiv-1 virus-like particles.
\newblock {\em The Journal of infectious diseases}, 210(1):99--110.

\bibitem[Huang et~al., 2022]{huang2022baseline}
Huang, Y., Zhang, Y., Seaton, K.~E., De~Rosa, S., Heptinstall, J., Carpp, L.~N., Randhawa, A.~K., McKinnon, L.~R., McLaren, P., Viegas, E., et~al. (2022).
\newblock Baseline host determinants of robust human hiv-1 vaccine-induced immune responses: A meta-analysis of 26 vaccine regimens.
\newblock {\em Ebiomedicine}, 84:104271.

\bibitem[Imai, 2008]{imai2008variance}
Imai, K. (2008).
\newblock Variance identification and efficiency analysis in randomized experiments under the matched-pair design.
\newblock {\em Statistics in medicine}, 27(24):4857--4873.

\bibitem[Imbens and Rubin, 2015]{imbens2015causal}
Imbens, G.~W. and Rubin, D.~B. (2015).
\newblock {\em Causal Inference in Statistics, Social, and Biomedical Sciences}.
\newblock Cambridge University Press.

\bibitem[Li et~al., 2018]{li2018asymptotic}
Li, X., Ding, P., and Rubin, D.~B. (2018).
\newblock Asymptotic theory of rerandomization in treatment--control experiments.
\newblock {\em Proceedings of the National Academy of Sciences}, 115(37):9157--9162.

\bibitem[Lipkovich et~al., 2023]{lipkovich2023overview}
Lipkovich, I., Svensson, D., Ratitch, B., and Dmitrienko, A. (2023).
\newblock Overview of modern approaches for identifying and evaluating heterogeneous treatment effects from clinical data.
\newblock {\em Clinical Trials}, page 17407745231174544.

\bibitem[Neyman and Iwaszkiewicz, 1935]{neyman1935statistical}
Neyman, J. and Iwaszkiewicz, K. (1935).
\newblock Statistical problems in agricultural experimentation.
\newblock {\em Supplement to the Journal of the Royal Statistical Society}, 2(2):107--180.

\bibitem[Neyman, 1923]{neyman1923application}
Neyman, J.~S. (1923).
\newblock {On the application of probability theory to agricultural experiments. Essay on principles. Section 9.}
\newblock {\em Annals of Agricultural Sciences}, 10:1--51.

\bibitem[Rosenbaum, 2017]{rosenbaum2017observation}
Rosenbaum, P. (2017).
\newblock {\em Observation and Experiment: An Introduction to Causal Inference}.
\newblock Harvard University Press.

\bibitem[Rosenbaum, 1987]{rosenbaum1987sensitivity}
Rosenbaum, P.~R. (1987).
\newblock Sensitivity analysis for certain permutation inferences in matched observational studies.
\newblock {\em Biometrika}, 74(1):13--26.

\bibitem[Rosenbaum, 2002]{rosenbaum2002observational}
Rosenbaum, P.~R. (2002).
\newblock {\em Observational Studies}.
\newblock Springer.

\bibitem[Rosenbaum, 2010]{rosenbaum2010design}
Rosenbaum, P.~R. (2010).
\newblock {\em Design of observational studies}, volume~10.
\newblock Springer.

\bibitem[Rosenbaum and Silber, 2009]{rosenbaum2009amplification}
Rosenbaum, P.~R. and Silber, J.~H. (2009).
\newblock Amplification of sensitivity analysis in matched observational studies.
\newblock {\em Journal of the American Statistical Association}, 104(488):1398--1405.

\bibitem[Rubin, 1974]{rubin1974estimating}
Rubin, D.~B. (1974).
\newblock Estimating causal effects of treatments in randomized and nonrandomized studies.
\newblock {\em Journal of Educational Psychology}, 66(5):688.

\bibitem[Rubin, 2005]{rubin2005causal}
Rubin, D.~B. (2005).
\newblock Causal inference using potential outcomes: Design, modeling, decisions.
\newblock {\em Journal of the American Statistical Association}, 100(469):322--331.

\bibitem[Sedransk and Meyer, 1978]{sedransk1978confidence}
Sedransk, J. and Meyer, J. (1978).
\newblock Confidence intervals for the quantiles of a finite population: simple random and stratified simple random sampling.
\newblock {\em Journal of the Royal Statistical Society: Series B (Methodological)}, 40(2):239--252.

\bibitem[Silber et~al., 2001]{silber2001matching}
Silber, J.~H., Rosenbaum, P.~R., Trudeau, M.~E., Even-Shoshan, O., Chen, W., Zhang, X., and Mosher, R.~E. (2001).
\newblock Multivariate matching and bias reduction in the surgical outcomes study.
\newblock {\em Medical Care}, 39(10):1048--1064.

\bibitem[Stuart et~al., 2011]{stuart2011use}
Stuart, E.~A., Cole, S.~R., Bradshaw, C.~P., and Leaf, P.~J. (2011).
\newblock The use of propensity scores to assess the generalizability of results from randomized trials.
\newblock {\em Journal of the Royal Statistical Society: Series A (Statistics in Society)}, 174(2):369--386.

\bibitem[Su and Li, 2023]{10.1093/biomet/asad030}
Su, Y. and Li, X. (2023).
\newblock {Treatment Effect Quantiles in Stratified Randomized Experiments and Matched Observational Studies}.
\newblock {\em Biometrika}.
\newblock asad030.

\bibitem[Wang, 2015]{Wang:2015}
Wang, W. (2015).
\newblock Exact optimal confidence intervals for hypergeometric parameters.
\newblock {\em Journal of the American Statistical Association}, in press.

\bibitem[Wu and Ding, 2021]{wu2021randomization}
Wu, J. and Ding, P. (2021).
\newblock Randomization tests for weak null hypotheses in randomized experiments.
\newblock {\em Journal of the American Statistical Association}, 116(536):1898--1913.

\bibitem[Zhang et~al., 2023]{zhang2023matching}
Zhang, B., Small, D.~S., Lasater, K.~B., McHugh, M., Silber, J.~H., and Rosenbaum, P.~R. (2023).
\newblock Matching one sample according to two criteria in observational studies.
\newblock {\em Journal of the American Statistical Association}, 118(542):1140--1151.

\end{thebibliography}

\end{document}